\documentclass[letterpaper, USenglish, cleveref, pdfa]{lipics-v2021}

\usepackage{amsmath,amssymb,amsfonts}
\usepackage{algorithmic}
\usepackage{graphicx}
\usepackage{textcomp}
\usepackage{xcolor}
\def\BibTeX{{\rm B\kern-.05em{\sc i\kern-.025em b}\kern-.08em
    T\kern-.1667em\lower.7ex\hbox{E}\kern-.125emX}}
\usepackage{cite}
\usepackage{graphicx}
\usepackage{textcomp}
\usepackage[numbers]{natbib}
\usepackage{underscore}
\usepackage{siunitx}
\usepackage{tikz}
\usepackage[linesnumbered,lined, vlined, ruled]{algorithm2e}

\SetCommentSty{mycommfont}
\usepackage{makecell}
\usepackage{subcaption}
\usepackage{tablefootnote}
\usepackage[shortlabels]{enumitem}
\usepackage{xcolor}
\usepackage{nameref}
\definecolor{ForestGreen}{rgb}{0.1333,0.5451,0.1333}
\definecolor{DarkRed}{rgb}{0.8,0,0}
\definecolor{Red}{rgb}{1,0,0}
\hypersetup{
linktocpage=true,
colorlinks,
linkcolor=DarkRed,citecolor=ForestGreen,
bookmarksopen,bookmarksnumbered}
\usepackage{multirow}
\usepackage{soul} 

\crefname{figure}{Figure}{Figures}

\newcommand{\calS}{\mathcal{S}}
\newcommand{\calT}{\mathcal{T}}

\newcommand{\abs}[1]{\vert #1 \vert}
\newcommand{\Rec}{\normalfont\text{Ret}}
\newcommand{\Par}{\normalfont\text{Par}}
\newcommand{\Count}{\normalfont\text{Count}}
\newcommand{\Anc}{\normalfont\text{Anc}}
\newcommand{\Dep}{\normalfont\text{Dep}}
\newcommand{\ExDep}{\normalfont\text{ExtDep}}
\newcommand{\OPT}{\normalfont\text{OPT}}
\newcommand{\ALG}{\normalfont\text{ALG}}
\newcommand{\MSR}{\textsc{MSR}}
\newcommand{\MMR}{\textsc{MMR}}
\newcommand{\BSR}{\textsc{BSR}}
\newcommand{\BMR}{\textsc{BMR}}
\newcommand{\NP}{\textsc{NP}}
\newcommand{\DTIME}{\textsc{DTIME}}
\newcommand{\calA}{\mathcal{A}}
\newcommand{\calR}{\mathcal{R}}

\newcommand{\DP}{\normalfont\text{DP}}

\newcommand{\DPBMR}{\textsf{DP-BMR}}
\newcommand{\DPMSR}{\textsf{DP-MSR}}
\newcommand{\LMG}{\textsf{LMG}}
\newcommand{\MP}{\textsf{MP}}
\newcommand{\LMGA}{\textsf{LMG-All}}

\newcommand{\minuseq}{\mathrel{-}=}

\newcommand{\bob}[1]{{\color{blue} Bob: #1}}
\newcommand{\sofia}[1]{{\color{magenta} sofia: #1}}
\newcommand{\pat}[1]{{\color{red} pat: #1}}
\newcommand{\samir}[1]{{\color{brown} samir: #1}}
\newcommand{\koyel}[1]{{\color{purple} Koyel: #1}}

\renewcommand{\bob}[1]{}
\renewcommand{\sofia}[1]{}
\renewcommand{\pat}[1]{}
\renewcommand{\samir}[1]{}
\renewcommand{\koyel}[1]{}

\pdfoutput=1 
\hideLIPIcs  

\title{To Store or Not to Store: a graph theoretical approach for Dataset
Versioning}

\titlerunning{To Store or Not to Store}

\author{Anxin Guo\footnote{These authors contributed equally to this work.}}{Computer Science Department, Northwestern University, Evanston, USA}{anxinbguo@gmail.com}{}{}

\author{Jingwei Li\footnotemark[1]{}}{Department of Industrial Engineering and Operations Research, Columbia University, New York City, USA}{jl6639@columbia.edu}{}{}

\author{Pattara Sukprasert}{Databricks, San Francisco, USA}{pattara.sk127@gmail.com}{}{}

\author{Samir Khuller}{Computer Science Department, Northwestern University, Evanston, USA}{samir.khuller@northwestern.edu}{}{}

\author{Amol Deshpande}{Department of Computer Science, University of Maryland, College Park, USA}{amol@umd.edu}{}{}

\author{Koyel Mukherjee}{Adobe Research, Bangalore, India}{komukher@adobe.com}{}{}

\authorrunning{A. Guo, J. Li, P. Sukprasert, S. Khuller, A. Deshpande, and K. Mukherjee} 

\Copyright{Anxin Guo, Jingwei Li, Pattara Sukprasert, Samir Khuller, Amol Deshpande, Koyel Mukherjee}

\ccsdesc[500]{Theory of computation~Approximation algorithms}
\ccsdesc[300]{Information systems~Data management systems}

\keywords{Version Control Systems, Data Management, Approximation Algorithm, Combinatorial Optimization.}

\relatedversion{} 

\acknowledgements{A. Guo, J. Li, P. Sukprasert, and S. Khuller were supported by a gift from Adobe Research. Significant part of this work was done while J. Li was an undergraduate student at Northwestern Univerty and while P. Sukprasert was a Ph.D. candidate at Northwestern University.}

\supplementdetails[linktext=GitHub]{Source code}{https://github.com/Sooooffia/Graph-Versioning/tree/main} 

\nolinenumbers 

\begin{document}

\maketitle

\begin{abstract}
In this work, we study the \emph{cost efficient data versioning problem}, where the goal is to optimize the storage and reconstruction (retrieval) costs of data versions, given a graph of datasets as nodes and edges capturing edit/delta information. One central variant we study is \textsc{MinSum Retrieval (MSR)} where the goal is to minimize the total retrieval costs, while keeping the storage costs bounded. This problem (along with its variants) was introduced by Bhattacherjee et al.~[VLDB'15]. While such problems are frequently encountered in collaborative tools (e.g., version control systems and data analysis pipelines), to the best of our knowledge, no existing research studies the theoretical aspects of these problems. 

We establish that the currently best-known heuristic, LMG (introduced in Bhattacherjee et al.~[VLDB'15]) can perform arbitrarily badly in a simple worst case. Moreover, we show that it is hard to get $o(n)$-approximation for MSR on general graphs even if we relax the storage constraints by an $O(\log n)$ factor. Similar hardness results are shown for other variants. Meanwhile, we propose poly-time approximation schemes for tree-like graphs, motivated by the fact that the graphs arising in practice from typical edit operations are often not arbitrary. As version graphs typically have low treewidth, we further develop new algorithms for bounded treewidth graphs. 

Furthermore, we propose two new heuristics and evaluate them empirically. First, we extend LMG by considering more potential ``moves'', to propose a new heuristic LMG-All. LMG-All consistently outperforms LMG while having comparable run time on a wide variety of datasets, i.e., version graphs. Secondly, we apply our tree algorithms on the minimum-storage arborescence of an instance, yielding algorithms that are qualitatively better than all previous heuristics for MSR, as well as for another variant \textsc{BoundedMin Retrieval} (BMR). 
\end{abstract}


\section{Introduction}\label{sec:intro}
\bob{TODO:
\begin{enumerate}
    \item Change repo to public repo. 
\end{enumerate}}

The management and storage of data versions has become increasingly important. As an example, the increasing usage of online collaboration tools allows many collaborators to edit an original dataset simultaneously, producing multiple versions of datasets to be stored daily. 
Large number of dataset versions also occur often in industry data lakes \cite{dataLake} where huge tabular datasets like product catalogs might require a few records (or rows) to be modified periodically, resulting in a new version for each such modification. Furthermore, in Deep Learning pipelines, multiple versions are generated from the same original data for training and insight generation. At the scale of terabytes or even petabytes, storing and managing all the versions is extremely costly in the aforementioned situations~\cite{koyel-ICDE}. Therefore, 
it is no surprise that data version control is emerging as a hot area in the industry \cite{git, pachyderm, DVC, TerminusDB, LakeFS, Dolt}, and even popular cloud solution providers like Databricks are now capturing data lineage information, which helps in effective data version management \cite{databricks}.

In a pioneering paper, Bhattacherjee et al.~\cite{versioningAmol} proposed a model capturing the trade-off between \textit{storage} cost and \textit{retrieval} (recreation) cost. 
The problems they studied can be defined as follows. Given dataset versions 
and a subset of the \emph{``deltas''} between them,
 find a compact representation that minimizes the overall storage as well as the retrieval costs of the versions. This involves a decision for each version:  either we \textit{materialize} it (i.e., store it explicitly) or we store a ``delta'' and rely on edit operations to retrieve the version from another materialized version if necessary. The downside of the latter is that, to retrieve a version that was not materialized, we have to incur a computational overhead that we call \textit{retrieval cost}. 

\cref{fig:version-graph}, 
taken from Bhattacherjee et al.\cite{versioningAmol}, illustrates the central point through different storage options. (i) shows the input graph, with annotated storage and retrieval costs 
. If the storage size is not a concern, we should store all versions as in (ii). From (iii) to (iv), it is clear that, by materializing $v_3$, we shorten the retrieval costs of $v_3$ and $v_5$.

\begin{figure}[t]
  \centering
  \includegraphics[width=0.8\columnwidth]{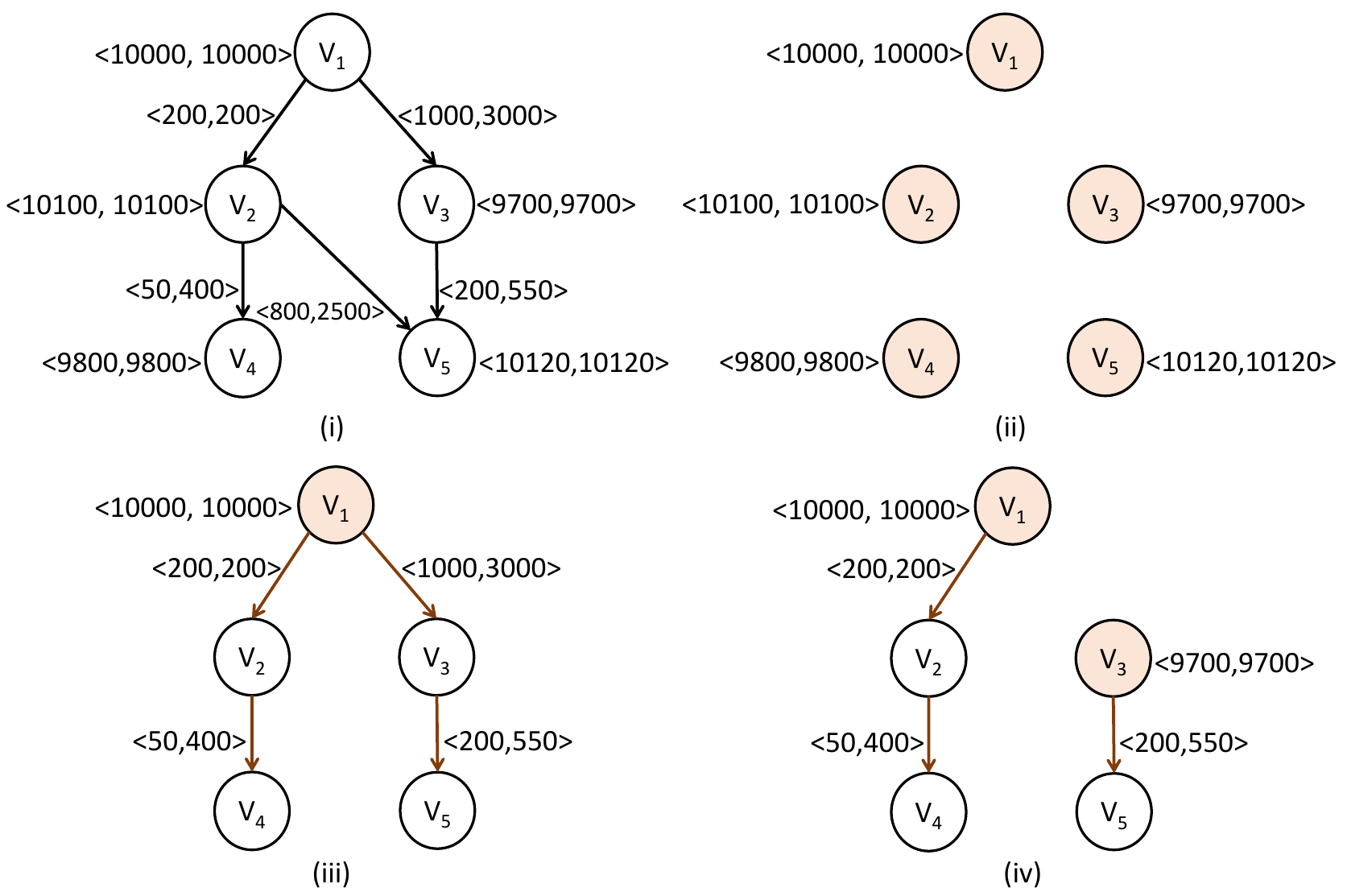}
 \caption{(i) A version graph over 5 datasets – annotation  ⟨a, b⟩  indicates
      a storage cost of $a$ and a retrieval cost of $b$;
     (ii, iii, iv) three possible storage graphs. The figure is taken from~\cite{versioningAmol}}
\label{fig:version-graph}
\end{figure}

This retrieval/storage trade-off leads to combinatorial problems of minimizing one type of cost, given a constraint on the other. Moreover, as an objective function, the retrieval cost can be measured by either the maximum or total (or equivalently average) retrieval cost of files. This yields four different optimization problems (Problems 3-6 in~\cref{table:problems1-6}). While the first two problems in the table are easy, the other four turn out to be NP-hard and hard to approximate, as we will soon discuss. 

\renewcommand{\arraystretch}{1.5}
\begin{table}[ht]
\begin{center}
    \footnotesize
    \begin{tabular}{|c||c|c|}
        \hline
        \textbf{Problem Name} & \textbf{Storage} & \textbf{Retrieval}\\
        \hline
        \hline
        \textsc{Minimum Spanning Tree}  & min & $R(v) < \infty,~\forall v$\\
        \hline
        \textsc{Shortest Path Tree}  & $< \infty$ & min $\{ \max_{v} R(v)\}$ \\
        \hline
        \textsc{MinSum Retrieval (MSR)} & $\leq \calS$ & min$\{ \sum_{v} R(v)\}$ \\
        \hline
        \textsc{MinMax Retrieval (MMR)} & $\leq \calS$ & min $\{ \max_{v} R(v)\}$ \\
        \hline
        \textsc{BoundedSum Retrieval (BSR)} & min & $\sum_v R(v) \leq \calR$ \\
        \hline
        \textsc{BoundedMax Retrieval (BMR)} & min & $ \max_v R(v) \leq \calR$ \\
        \hline
    \end{tabular}
\end{center}
\caption{Problems 1-6. Here, $R(v)$ is the retrieval cost of version $v$, while $\calR,\calS$ are the retrieval and storage constraints, respectively.}
\label{table:problems1-6}
\end{table}



There are some follow-up works on this model~\cite{simpleFollow-up,Derakhshan22material,huang20OrpheusDB}. However, those either formulate new problems in different use cases~\cite{Derakhshan22material,manne22CHEX,huang20OrpheusDB} or implement a system incorporating the feature to store specific versions and deltas~\cite{huang20OrpheusDB,Wang18forkbase,schule19Versioning}. We will discuss this in more detail in \cref{subsec:related-works}.

\subsection{Our Contributions}

We provide the first set of \textit{approximation algorithms} and \textit{inapproximability results} for the aforementioned optimization problems under various conditions. Our theoretical results also give rise to practical algorithms which perform very well on real-world data. 
\begin{table}[h]
\begin{center}
    \footnotesize
    \begin{tabular}{|c|c|c|c|c|}
        \hline
        \textbf{Problem} & \textbf{Graph type} & \textbf{Assumptions} & \textbf{Inapproximability} \\
        \hline
        \hline
        \multirow{3}{*}{MSR} & arborescence & \multirow{8}{*}{\makecell{Triangle inequality \\ $r=s$ on edges\footnote{}}} & 1 \\\cline{4-4}\cline{2-2}
        
        & undirected && $1+\frac{1}{e}-\epsilon$ \\\cline{4-4}\cline{2-2}
        
        & general && $\Omega(n)$\footnote{This is true even if we relax $\calS$ by $O(\log n)$.}  \\\cline{1-2}\cline{4-4}
        
        \multirow{2}{*}{MMR} & undirected && $2-\epsilon$ \\\cline{4-4}\cline{2-2}
        
        & general && $\log^*n - \omega(1)$ \\\cline{1-2}\cline{4-4}
        
        \multirow{2}{*}{BSR} & arborescence && 1 \\\cline{4-4}\cline{2-2}
        
        & undirected && $(\frac{1}{2}-\epsilon)\log n$ \\\cline{1-2}\cline{4-4}
        
        BMR & undirected && $(1-\epsilon)\log n$ \\
        \hline
    \end{tabular}
\end{center}
\caption{Hardness results}\label{table:Hardness}
\end{table}

\footnotetext[2]{Both are assumptions in previous work~\cite{versioningAmol} that simplify the problems. We note that our algorithms function even without these assumptions.
}
\footnotetext[3]{This is true even if we relax $\calS$ by $O(\log n)$.}

\subparagraph*{\MMR{} and \BMR{}.} In \cref{sec:hardness} we prove that it is hard to approximate \MMR{} within $\log^* n$\footnote{$\log^* n$ is ``iterated logarithm'', defined as the number of times we iteratively take logarithm before the result is at most $1$.} factor and \BMR{} within $\log n$ factor on general inputs. Meanwhile, in \cref{sec:MMRBMR} we give a polynomial-time dynamic programming (DP) algorithm for the two problems on \textit{bidirectional trees}, i.e., digraphs whose underlying undirected graph\footnote{The undirected graph formed by disregarding the orientations of the edges.} is a tree. These inputs capture the cases where new versions are generated via edit operations. 

We also briefly describe an FPTAS (defined below) for \MMR{}, analogous to the main result for \MSR{} in \cref{sec:fptas}. 

\subparagraph*{\MSR{} and \BSR{}.} In \cref{sec:hardness} we prove that it is hard to design $\big(O(n),O(\log n)\big)$-bicriteria approximation\footnote{An $(\alpha,\beta)$-bicriteria approximation refers to an algorithm that potentially exceeds the constraint by $\alpha$ times, in order to achieve a $\beta$-approximation of the objective. See \cref{sec:prelim} for an example.}
for \MSR{} or $O(\log n)$-approximation for \BSR{}. It is also NP-hard to solve the two problems exactly on trees. 

On the other hand, we again use DP to design a fully polynomial-time approximation scheme (FPTAS) for \MSR{} on \textit{bounded treewidth graphs}. These inputs capture many practical settings: bidirectional trees have width 1, series-parallel graphs have width 2, and the GitHub repositories we use in (\cref{sec:experiments}) all have low treewidth.\footnote{\texttt{datasharing}, \texttt{styleguide}, and \texttt{leetcode} have treewidth 2,3, and 6 respectively.}

\subparagraph*{New Heuristics.} 
We improved \LMG{} into a more general \LMGA{} algorithm for solving \MSR{}. \LMGA{} outperforms \LMG{} in all our experiments and runs faster than \LMG{} on sparse graphs. 

Inspired by our algorithms on trees, we also propose two DP heuristics for MSR and BMR. Both algorithms perform extremely well 
even when the input graph is not tree-like. Moreover, there are known procedures for parallelizing general DP algorithms~\cite{stivala2010lock}, so our new heuristics are potentially more practical than previous ones, which are all sequential. 

\begin{table}[t]
\begin{center}
    \footnotesize
    \begin{tabular}{|l|c|c|c|}
        \hline
        \textbf{ Graphs} & \textbf{Problems} & \textbf{Algorithm} & \textbf{Approx.}  \\
        \hline
        \hline
        General Digraph & MSR & \textsf{LMG-All} & heuristic  \\
        \hline
        
        \multirow{2}{*}{Bounded Treewidth} & MSR \& MMR & \multirow{2}{*}{\textsf{DP-BTW}} & $1+\epsilon$ \\
        
        \cline{2-2}\cline{4-4}
        & BSR \& BMR & & $(1,1+\epsilon)$  \\
        
        \hline
        \multirow{2}{*}{Bidirectional Tree} & MMR & \multirow{2}{*}{\textsf{DP-BMR}} & \multirow{2}{*}{exact} \\
        \cline{2-2}
        & BMR & & \\
        \hline
    \end{tabular}
\end{center}
\caption{Algorithms summary. }\label{table:Algorithms}
\end{table}

\subsection{Related Works}
\label{subsec:related-works}
\subsubsection{\textbf{Theory}}
There was little theoretical analysis of the exact problems we study. 
The optimization problems are first formalized in Bhattacherjee et al.~\cite {versioningAmol}, which also compared the effectiveness of several proposed heuristics on both real-world and synthetic data.
 Zhang et al.\cite{simpleFollow-up} followed up by considering a new objective that is a weighted sum of objectives in MSR and MMR. They also modified the heuristics to fit this objective. 
There are similar concepts, including \emph{Light Approximate Shortest-path Tree (LAST)}~\cite{khuller93LAST} and \emph{Shallow-light Tree (SLT)}~\cite{kortsarz1997approximating, hajiaghayi2009approximating, HaeuplerEmbedding, marathe1998bicriteria, RRaviBroadcast, RezaBuyAtBulk}. 
However, this line of work focuses mainly on undirected graphs and their algorithms do not generalize to the directed case.
Among the two problems mentioned, SLT is closely related to MMR and BMR. 
Here, the goal is to find a tree that is \textbf{light} (minimize weight) and \textbf{shallow} (bounded depth).
To our knowledge, there are only two works that give approximation algorithms for directed shallow-light trees. Chimani and Spoerhase
\cite{chimani15NDesign} give a bi-criteria $(1+\epsilon, n^\epsilon)$-approximation algorithm that runs in polynomial-time.
Recently, Ghuge and Nagarajan~\cite{ghuge2022quasi} gave a $O(\frac{\log n}{\log\log n})$-approximation algorithm for \textit{submodular tree orienteering} that runs in quasi-polynomial time.
Their algorithm can be adapted
into $O(\frac{\log^2 n}{\log\log n})$-approximation for \BMR{}.
For \textsc{MSR}, 
their algorithm gives $\big ( O(\frac{\log^2 n}{\log\log n})  , O(\frac{\log^2 n}{\log\log n})  \big )$-approximation. The idea is to run their algorithm for many rounds, where the objective of each round is to \textit{cover as many nodes as possible}. 
\subsubsection{\textbf{Systems}} 
To implement a system captured by our problems, components spanning multiple lines of works are required.
For example, to get a graph structure, one has to keep track of history of changes.
This is related to the topic of data provenance~\cite{buneman00DataProvenance,simmhan2005survey}.
Given a graph structure, the question of modeling ``deltas'' is also of interest.
There is a line of work dedicated to studying how to implement \texttt{diff} algorithms in different contexts~\cite{hunt98Delta,burn98inplace,xia14Ddelta,macdonald2000file,suel2002zdelta}.


In the more flexible case, one may think of creating deltas without access to the change history. However, computing all possible deltas is too wasteful, hence it is necessary to utilize other approaches to identify similar versions/datasets.
Such line of work is known as dataset discovery or dataset similarlity~\cite{dataLake,jayawardana19DFS,fernandez18Aurum,brickley19GoogleDataSet,bogatu20DDD}.

Several follow-up works of Bhattacherjee et al.\cite{versioningAmol} have implemented systems with a feature that saves only selected versions to reduce redundancy. 
There are works focusing on version control for relational databases~\cite{brown21DSDB,huang20OrpheusDB,schule19Versioning,Wang18forkbase,chavan17Dex,maddox16Decibel,bhardwaj14DataHub,seering12EffVer} and works focusing on graph snapshots~\cite{Ying20pensive,khurana12EffSnap,manne22CHEX}.
However, since their focus was on designing full-fledged systems, 
the algorithms they proposed are rather simple heuristics with no theoretical guarantees.


\subsubsection{\textbf{Usecases}}
In a version control system such as git, our problem is similar to what \texttt{git pack} command aims to do.\footnote{https://www.git-scm.com/docs/git-pack-objects}
The original heuristic for \texttt{git pack}, as described in an IRC log, is to sort objects in particular order and only create deltas between objects in the same window.\footnote{https://github.com/git/git/blob/master/Documentation/technical/pack-heuristics.txt} It is shown in Bhattacherjee et al.~\cite{versioningAmol} that git's heuristic does not work well compared to other methods.

SVN, on the other hand, only stores the most recent version and the deltas to past versions~\cite{nagel2006subversion}. Other existing data version management systems include~\cite{pachyderm, DVC, TerminusDB, LakeFS, Dolt}, which offer git-like capabilities suited for different use cases, such as data science pipelines in enterprise setting, machine learning-focused, data lake storage, graph visualization, etc.

Though not directly related to our work, recently, there has been a lot of work exploring algorithmic and systems related optimizations for reducing storage and maintenance costs of data. For example, Mukherjee et al.~\cite{koyel-ICDE} proposes optimal multi-tiering, compression and data partitioning, along with predicting access patterns for the same. Other works that exploit multi-tiering to optimize performance include e.g., \cite{hermes, devarajan2020hcompress, devarajan2020hfetch, cheng2015cast} and/or costs, e.g., \cite{cheng2015cast, ERRADI2020110457, liu2019transfer, liu2021keep, si2022cost, kinoshita2021cost}. Storage and data placement in a workload aware manner, e.g., \cite{anwar2015taming, anwar2016mos, cheng2015cast} and in a device aware manner, e.g., \cite{vogel2020mosaic, lasch2022cost, lasch2021workload} have also been explored. \cite{devarajan2020hcompress} combine compression and multi-tiering for optimizing latency.

\section{Preliminaries}
\label{sec:prelim}
In this section, the definition of the problems, notations, simplifications, and assumptions will be formally introduced.

\subsection{Problem Setting}

In the problems we study, we are given a directed \textit{version graph} $G=(V,E)$, 
where vertices represent \emph{versions} and edges capture the \emph{deltas} between versions.
Every edge $e = (u,v)$ is associated with two weights: 
storage cost $s_e$ and retrieval cost $r_e$.\footnote{If $e=(u,v)$, we may use $s_{u,v}$ 
in place of $s_e$ and $r_{u,v}$ in place of $r_e$.} The cost of storing $e$ is $s_e$, and it takes $r_e$ time to retrieve $v$ once we retrieved $u$. Every vertex $v$ is associated with only the storage cost, $s_v$, of storing (materializing) the version. Since there is usually a smallest unit of cost in the real world, we will assume $s_v,s_e,r_e \in \mathbb{N}$ for all $v\in V,e\in E$.

To retrieve a version $v$ from a materialized version $u$, there must be some path $P=\{(u_{i-1},u_i)\}_{i=1}^n$ with $u_0 = u, u_n = v$, such that all edges along this path are stored. In such cases, we say that $v$ is retrieved from materialized $u$ with retrieval cost $R(v) = \sum_{i=1}^n r_{(u_{i-1},u_i)}$. In the rest of the paper, we say $v$ is ``retrieved from $u$'' if $u$ is in the path to retrieve $v$, and $v$ is ``retrieved from materialized $u$'' if in addition $u$ is materialized. 

The general optimization goal is to select vertices $M \subseteq V$ and edges $F \subseteq E$ of \textit{small} size (w.r.t. storage cost $s$), such that for each $v \in V \setminus M$, there is a \textit{short} path (w.r.t retrieval cost $r$) from a materialized vertex to $v$. 
Formally, we want to minimize (a) total storage cost $\sum_{v\in M} s_v + \sum_{e \in F} s_e$, and (b) total (resp. maximum) retrieval cost $\sum_{v \in V} R(v)$ (resp. $\max_{v \in V} R(v)$). 

Since the storage and retrieval objectives are negatively correlated, 
a natural problem is to constrain one objective and minimize the other. 
With this in mind, four different problems are formulated, as described by Problems 3-6 in \cref{table:problems1-6}. These problems are originally defined in Bhattacherjee et al.~\cite{versioningAmol}, although we use different names for brevity. Since the first two problems are well studied, we do not discuss them further.

\subsection{Further Definitions}\label{subsec:further-assumptions}
We hereby formalize several simplifications and complications, to capture more realistic aspects of the problem. Most of the proposed variants are natural and considered by Bhattacherjee et al.~\cite{versioningAmol}. 

\subparagraph*{Triangle inequality:} It is natural to assume that both weights satisfy triangle inequality, i.e., $r_{u,v} \leq r_{u,w} + r_{w,v}$, since we can always implement the delta $r_{u,v}$ by implementing first $r_{u,w}$ and then $r_{w,v}$.

In fact, a more general triangle inequality should hold when we consider the materialization costs $s_v$, as it's often true that $s_u + s_{u,v} \geq s_v$ for all pairs of $u,v\in V$. 

All hardness results in this paper hold under the generalized triangle inequality. 

\subparagraph*{Directedness:} It is possible that for two versions $u$ and $v$, $r_{u,v} \neq r_{v,u}$. In real world, deletion is also significantly faster and easier to store than addition of content. Therefore, Bhattacherjee et al.~\cite{versioningAmol} considered both directed and undirected cases; we argue that it is usually more natural to model the problems as directed graphs and focus on that case. Note that in the most general directed setting, it is possible that we are given the delta $(u,v)$ but not $(v,u)$. (For our purposes, this is equivalent to having a worse-than-trivial delta, with $s_{v,u} \geq s_u$.)

Directed and undirected cases are considered separately in our hardness results, and all our algorithms apply in the more general directed case. 

\subparagraph*{Single weight function:} This is the special case where the storage cost function $s_e$ and retrieval cost $r_e$ function are identical up to some scaling. This can be seen in the real world, for example, when we use simple \texttt{diff} to produce deltas. We note that the \textit{random compression} construction in our experiments (\cref{sec:experiments}) is designed to simulate two distinct weight functions. 

All our hardness results hold for single weight functions. All our approximation algorithms work even when the two weight functions are very different. 

\subparagraph*{Arborescence and trees:}
An \emph{arborescence}, or a directed spanning tree, is a connected digraph where all vertices except a designated root have in-degree 1, and the root has in-degree 0. If each version is a modification on top of another version, then the ``natural" deltas automatically form an arboreal input instance.\footnote{This does not hold true for version controls because of the \texttt{merge} operation.} 
For practical reasons, we also consider \textit{bidirectional tree} instances, meaning that both $(u,v)$ and $(v,u)$ are available deltas with possibly different weights. Empirical evidence shows that having deltas in both directions can greatly improve the quality of the optimal solution.\footnote{Although not presented in this paper, we noticed that the minimum arborescences on all our experimental datasets tend to have much worse optimal costs, compared to the minimum bidirectional trees.}

\subparagraph*{Bounded treewidth:} At a high level, treewidth measures how similar a graph is to a tree~\cite{TreeWidthDefinition}. As one notable class of graphs with bounded treewidths, series-parallel graphs highly resemble the version graphs we derive from real-world repositories. Therefore, graphs with bounded treewidth is a natural consideration with high practical utility. 
We give precise definitions of this special case in \cref{subsec:DPBoundedTW}.

\medskip
We note that once we have an algorithm for MSR (resp. MMR), we can turn it into an algorithm for BSR (resp. BMR) by binary-searching over the possible values of the constraint. Due to the somewhat exchangeable nature of the storage and constraints in these problems, it's worth considering $(\alpha,\beta)$-bicriteria approximations, where we relax the constraint by some $\beta$ factor in order to achieve a $\alpha$-approximation. 
For example, an algorithm $A$ is $(\alpha, \beta)$-bicriteria approximation for \textsc{MSR} if it outputs a feasible solution with storage cost $\leq \alpha \cdot \mathcal S$ and retrieval cost $\leq \beta \cdot OPT$ where $OPT$ is the retrieval cost of an optimal solution.

\section{Hardness Results}\label{sec:hardness}
We hereby prove the main hardness results of the problems. For completeness, we define the notion of approximation algorithms, as used in this paper, in \cref{appendix:approx-alg-def}. We also include in \cref{appendix:known-hardness-problems} a list of well-studied optimization problems that are used in this section for reduction purposes. 

\subsection{Heuristics can be Arbitrarily Bad}\label{subsec:LMG-Bad}
First, we consider the approximation factor of the best heuristic for \MSR{} in Bhattacherjee et al.~\cite{versioningAmol}, Local Move Greedy (\LMG{}). 
The gist of this algorithm is to start with the arborescence that minimizes the storage cost, and 
iteratively materialize a version that most efficiently reduces retrieval cost per unit storage. In other words, in each step, a version is materialized with maximum $\rho$, where 

$$\rho = \frac{\text{
reduction in total of retrieval costs}}{\text{increase in storage cost}}.$$

We provide the pseudocode for \LMG{} in \cref{alg:LMG}. \pat{We should tidy the notation in the algorithm}

\begin{algorithm}[ht!]\caption{\textsc{Local Move Greedy} (\LMG{})}
    \label{alg:LMG} 
    \SetKwInOut{Input}{Input}\SetKwInOut{Output}{Output}
    \Input{ Version graph $G$, storage constraint $\calS$}
    \tcc{Constructing extended version graph with auxiliary root $v_{aux}$.}
    $V \gets V \cup \{ v_{aux} \}$\;
    \For{$v \in V \setminus \{v_{aux}\}$}{
        $E\gets E \cup \{ (v_{aux}, v) \}$\;
        $r_{(v_{aux}, v)} = 0$\;
        $s_{(v_{aux}, v)} = s_v$\;
    }
    Let $G_{aux} = (V, E)$\;
    \tcc{The main algorithm.}
    $T\gets$ minimum arborescence of $G_{aux}$ rooted at $v_{aux}$ w.r.t. weight function $s$\;
    Let $S(T)$ be the total storage cost of $T$\;
    Let $R(v)$ be the retrieval cost of $v$ in $T$\;
    Let $P(v)$ be the parent of $v$ in $T$\;
    $U\gets V$\;
    \While{$S(T) < \calS$}{
        $(\rho_{max},v_{max})\gets (0,\varnothing)$\;
        \For{$v\in U$ with $S(T)+s_v - s_{P(v),v}\leq \calS$}{
            $T'\gets T\setminus \{(P(v),v)\} \cup \{(v_{aux},v)\}$\;
            $\Delta = \sum_v \big( R(v) - R_{T'}(v) \big)$\;
            \If{$\Delta/(s_v - s_{P(v),v}) > \rho_{max}$}{
                $\rho_{max}\gets \Delta/(s_v - s_{P(v),v})$\;
                $v_{max}\gets v$\;
            }
        }
        $T\gets T\setminus \{(P(v_{max}),v_{max})\} \cup \{(v_{aux},v_{max})\}$\;
        $U\gets U\setminus\{v_{max}\}$\;
        \If{$U = \varnothing$}{
            \Return $T$\;
        }
    }
    \Return $T$\; 
\end{algorithm}


\begin{theorem}\label{thm:LMGBad}
\LMG{} has an arbitrarily bad approximation factor for \textsc{MinSum Retrieval}, even under the following assumptions:
\begin{romanenumerate}
 \item $G$ is a directed path; 
 \item there is a single weight function; and 
 \item triangle inequality holds.
\end{romanenumerate}
\end{theorem}

\begin{proof}
Consider the following chain of three nodes; the storage costs for nodes and the storage/retrieval costs for edges are labeled in \cref{fig:LMG-counter} Let $a$ be large and $\epsilon = b/c$ be arbitrarily small. To save space, we do not show $v_{aux}$ but only the nodes of the version graph.

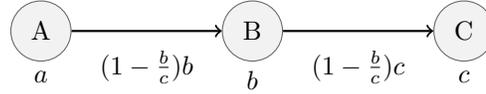
\begin{figure}
    \centering
    \begin{tikzpicture}[
    x = 80, y = 0, 
    roundnode/.style={circle, draw=black!60, fill=black!5, minimum size=8mm}, 
    scale=1]

      \node[roundnode] (A) at (0,0) [label=below:$a$, draw] {A};
      \node[roundnode] (B) at (1,0) [label=below:$b$,draw] {B};
      \node[roundnode] (C) at (2,0) [label=below:$c$,draw] {C};
      \draw [->, thick] (A) -- (B) node [midway, label=below:{$(1-\frac{b}{c})b$}] {};
      \draw [->, thick] (B) -- (C) node [midway, label=below:{$(1-\frac{b}{c})c$}] {};
    \end{tikzpicture}
    \caption{An adversarial example for \LMG{}.}
    \label{fig:LMG-counter}
\end{figure}

It is easy to check that triangle inequality holds on this graph. 

 In the first step of \LMG{}, the minimum storage solution of the graph is $\{A,(A,B),(B,C)\}$ with storage cost $a + (1-\epsilon)b + (1-\epsilon)c$.

Next, in the greedy step, two options are available: 
\begin{bracketenumerate}
\item \label{proof:LMGbad:opt1} Choosing $B$ and delete $(A,B)$: $$\rho_1 = \frac{2(1-\epsilon)b}{\epsilon b} =  \frac{2}{\epsilon}-1.$$
\item \label{proof:LMGbad:opt2} Choosing $C$ and delete $(B,C)$: $$\rho_2 = \frac{(1-\epsilon)b+(1-\epsilon)c}{\epsilon c} = \frac{(1-\epsilon)b}{b}+\frac{1-\epsilon}{\epsilon} = \frac{1}{\epsilon} -\epsilon < \frac{2}{\epsilon}-1.$$
\end{bracketenumerate}
With any storage constraint in range $\big[a+ (1-\epsilon)b + c, a+b+c \big)$, \LMG{} will choose Option~\ref{proof:LMGbad:opt1} which gives a total retrieval cost of $(1-\epsilon)c$. Note that with $\calS < a+b+c$, \LMG{} is not able to pick Option~\ref{proof:LMGbad:opt2} after taking Option~\ref{proof:LMGbad:opt1}. However, by choosing Option~\ref{proof:LMGbad:opt2}, which is also feasible, the total retrieval cost is $(1-\epsilon)b$. The proof is finished by observing that $c/b$ can be arbitrarily large.

\end{proof}

\subsection{Hardness Results on General Graphs} \label{subsec:hardnessResultsGeneral}
Here, we show the various hardness of approximations on general input graphs. We first focus on \textsc{MSR} and \textsc{MMR} where the constraint is on storage cost and the objective is on the retrieval cost. We then shift our attention to \textsc{BMR} and \textsc{BSR} in which the constraint is of retrieval cost and the objective function is on minimizing storage cost.

\subsubsection{Hardness for MSR and MMR}\hfill

\label{subsubsec:MSRMMRhardness}
\begin{theorem}
On version graphs with $n$ nodes, even assuming single weight function and triangle inequality, there is no: 
\begin{romanenumerate}
    \item $(\alpha,\beta)$-approximation for \textsc{MinSum Retrieval} if $\beta\leq \frac{1}{2}(1-\epsilon)\big(\ln n-\ln\alpha-O(1)\big)$; in particular, for some constant $c$, there is no $(c\cdot n)$-approximation without relaxing storage constraint by some $\Omega(\log n)$ factor, unless $\textsc{NP}\subseteq \textsc{DTIME}(n^{O(\log\log n)})$;

    \item $(1+\frac{1}{e}-\epsilon)$-approximation for \textsc{MinSum Retrieval} on undirected graphs for all $\epsilon>0$, unless $\textsc{NP}\subseteq \textsc{DTIME}(n^{O(\log\log n)})$;
    
    \item $\big ( \log^*(n)-\omega(1)\big )$-approximation for \textsc{MinMax Retrieval}, unless $\textsc{NP}\subseteq \textsc{DTIME}(n^{O(\log\log n)})$;

    \item $(2-\epsilon)$-approximation for \textsc{MinMax Retrieval} on undirected graphs for all $\epsilon > 0$, unless $\NP=$P. 
\end{romanenumerate}
\end{theorem}

\begin{proof}
\textbf{\textsc{MSR.} } There is an approximation-preserving (AP) reduction\footnote{In particular, a strict reduction. See, e.g., Crescenzi's note \cite{AP-Reduction} for more detail.} 
from \textsc{(Asymmetric) k-median} to \MSR{}. Let $s_{u,v}=r_{u,v}=d_{u,v}$, the distance from $u$ to $v$ in a (asymmetric) $k$-median instance. By setting the size of each version $v$ to some large $N$ and storage constraint to be $\calS=kN+n$, we can restrict the instance to materialize at most $k$ nodes and retrieve all other nodes through deltas. For large enough $N$, an $(\alpha,\beta)$-approximation for \MSR{} provides an $(\alpha,\beta)$-approximation for \textsc{(Asymmetric) k-median}, just by outputting the materialized nodes. The desired results follow from known hardness for asymmetric~\cite{unpublishedArcher} or symmetric (see \cref{subsec:hard-problems}) \textsc{k-median}. 

\textbf{\textsc{MMR.}} A similar AP reduction exists from \textsc{(Asymmetric) k-center} to \MMR{}. Again, we can set all materialization costs to $N$ and $c_{u,v}=r_{u,v}=d_{u,v}$, and the desired result follows from the hardness of asymmetric~\cite{k-center-hardness} and symmetric~\cite{SymmetrickCenterHardness} \textsc{k-center}.
\end{proof}

\subsubsection{Hardness for BSR and BMR} \hfill

\begin{theorem}\label{thm:hardness}
On both directed and undirected version graphs with $n$ nodes, even assuming single weight function and triangle inequality, there is no:
\begin{romanenumerate}
    \item $(c_1\ln n)$-approximation for \textsc{BoundedSum Retrieval} for any $c_1<0.5$;
    
    \item $(c_2\ln n)$-approximation for \textsc{BoundedMax Retrieval} for any $c_2 < 1$. 
\end{romanenumerate}
unless $\textsc{NP}=\textsc{P}$. 
\end{theorem}
To prove this theorem, we will present our reduction to these two problems from \textsc{Set Cover}.
We then show their structural properties on \cref{lem:postProcessSolution,lem:postProcessSolution2}.
We finally show the proof at the end of this section.

\subparagraph*{Reduction.} Given a set cover instance with sets $A_1,\ldots,A_m$ and elements $o_1,\ldots,o_n$, we construct the following version graph:
\begin{enumerate}
    \item Build versions $a_i$ corresponding to $A_i$, and $b_j$ corresponding to $o_j$. All versions have size $N$ for some large $N\in\mathbb{N}$. 
    \item For all $i,j\in[m],i\neq j$, create symmetric delta $(a_i,a_j)$ of weight $1$. For each $o_j\in A_i$, create symmetric delta $(a_i,b_j)$ of weight $1$. 
\end{enumerate}


\begin{lemma}[BMR's structure]
\label{lem:postProcessSolution}
Assume we are given an approximate solution to \BMR{} on the above instance under max retrieval constraint $\calR=1$. In polynomial time, we can produce another feasible solution with equal or smaller total storage cost such that only the set versions are materialized. i.e., all $\{b_j\}_{j=1}^n$ are retrieved via deltas.


\end{lemma}
\begin{proof}[Proof of \cref{lem:postProcessSolution}]
    \pat{Looks incomplete, double check}
    Suppose some algorithm produces a solution that materializes $b_j$. 
    \subparagraph*{Case~1:}
    If there exists $a_i$ that needs to be retrieved through $b_j$
    (i.e., $o_j\in A_i$), 
    then we can replace the materialization of $b_j$ with that of $a_i$ and replace edges
    of the form $(b_j, a_k)$ with $(a_i, a_k)$.
    It is straightforward to see that neither storage cost nor retrieval cost increased in this process.
    
    \subparagraph*{Case~2:} If no other node is dependent on $b_j$, we can pick any $a_i$ such that $(a_i,b_j)$ exists (again, $o_j \in A_i$). If $a_i$ is already materialized in the original solution, then we can store $(a_i,b_j)$ instead of materializing $b_j$, which decreases storage cost. 
    
    \subparagraph*{Case~3:} If no $a_i$ adjacent to $b_j$ is materialized in the original solution, then some delta $(a_{i'},a_i)$ has to be stored with the former materialized to satisfy the $\calR=1$ constraint. We can hence materialize $a_i$, delete the delta $(a_{i'},a_i)$, and again replace the materialization of $b_j$ with the delta $(a_i,b_j)$ without increasing the storage. \cref{fig:hardness6} illustrates this case. 
\end{proof}

\begin{figure}[ht!]
  \centering
  \includegraphics[width=.7\columnwidth]{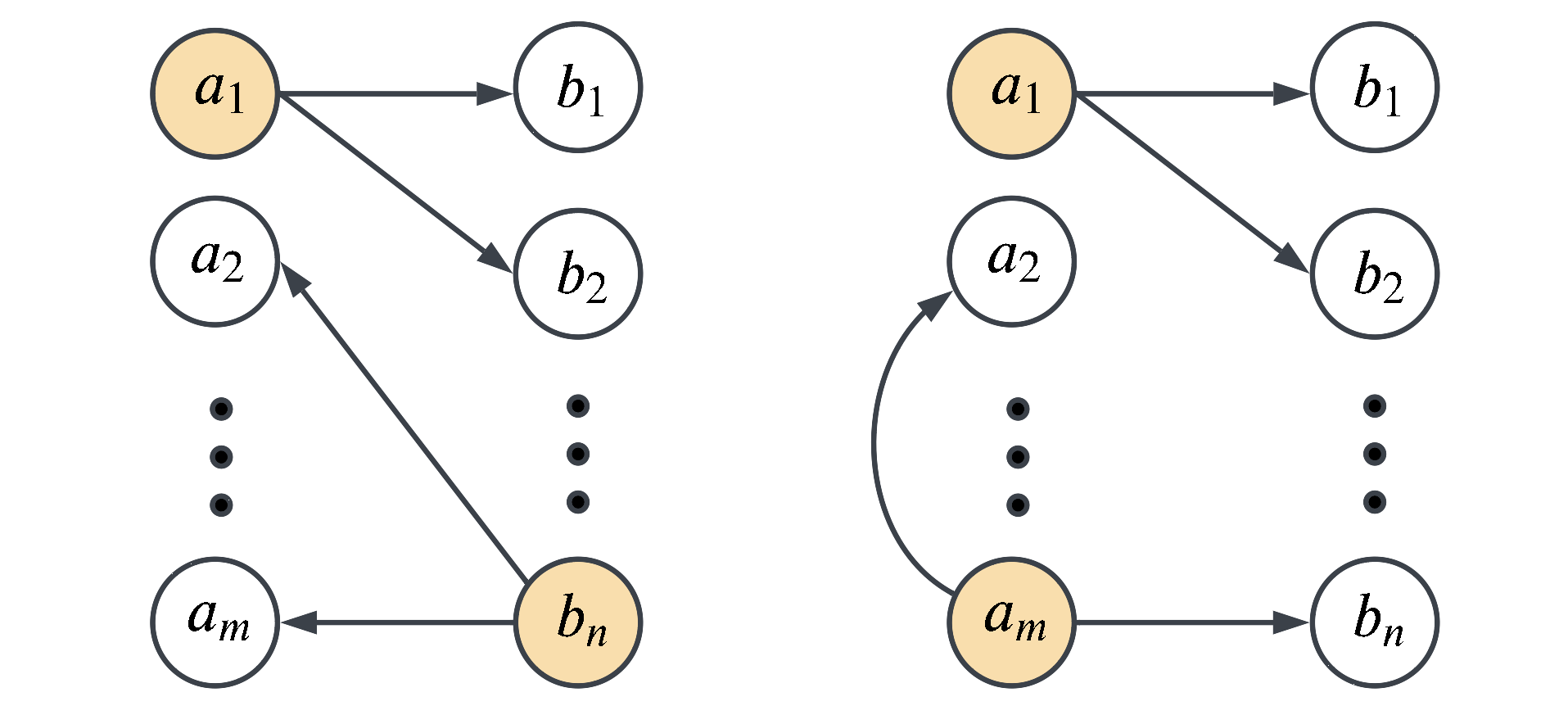}
    \caption{Case 3 in the proof of \cref{lem:postProcessSolution}. The improved solution is on the right. }
     \label{fig:hardness6}
\end{figure}

\begin{lemma}[BSR's structure]
\label{lem:postProcessSolution2}
Assume we are given an approximate solution to \BSR{} on the above version graph under total retrieval constraint $\calR = m-m_{\OPT}+n$, where $m_{\OPT}$ is the size of the optimal set cover. In polynomial time, we can produce another feasible solution to \BSR{} with equal or lower total storage cost, such that only the set versions are materialized. i.e., all $\{b_j\}_{j=1}^n$ are retrieved via deltas.
\end{lemma}
\begin{proof}[Proof of Lemma~\ref{lem:postProcessSolution2}]

We refer to the same three cases as in \cref{lem:postProcessSolution}. Similarly, if we have a solution where some $b_j$ is materialized, 

\subparagraph*{Case 1:} if some $a_i$ is retrieved through $b_j$, we can apply the same modification as \cref{lem:postProcessSolution}. We can replace the materialization of $b_j$ with that of $a_i$, and replace edges of the form $(b_j,a_k)$ with $(a_i,a_k)$. Neither the storage nor the retrieval cost increases in this case. 

Now, WLOG we assume no deltas $(b_j,a_i)$ are chosen. 

\subparagraph*{Case 2:} if no $a_i$ is retrieved through $b_j$, and some $a_i$ adjacent to $b_j$ is materialized, then method in \cref{lem:postProcessSolution} needs to be modified a bit in order to remove the materialization of $b_j$. If we simply retrieve $b_j$ via the delta $(a_i,b_j)$, we would lower the storage cot by $N-1$ and \textit{increase} the total retrieval cost by $1$. This renders the solution infeasible if the total retrieval constraint is already tight. 

To tackle this, we analyze the properties of the solutions  with total retrieval cost exactly $\calR$. Observe that all solutions must materialize at least $m_{\OPT}$ nodes at all time, so a configuration exhausting the constraint $R$ must have some version $w$ with retrieval cost at least $2$. If this $w$ is a set version, we can loosen the retrieval constraint by storing a delta of cost 1 from some materialized set instead. If $w$ is an element version, then we can materialize its parent version (a set covering it), which increases storage cost by $N-1$ and decreases total retrieval cost by at least 2. 

Either case, by performing the above action if necessary, we can resolve case 2 and obtain a approximate solution that is not worse than before. 

\subparagraph*{Case 3:} this is where each $a_i$ adjacent to $b_j$ neither retrieves through $b_j$ nor is materialized. Fix an $a_i$, then some delta $(a_{i'},a_i)$ has to be stored to retrieve $a_i$; WLOG we can assume that the former is materialized. We can thus materialize $a_i$, delete the delta $(a_{i'},a_i)$, and again replace the materialization of $b_j$ with the delta $(a_i,b_j)$ with no increase in either costs. 
\end{proof}

Equipped with \cref{lem:postProcessSolution} and \cref{lem:postProcessSolution2}, we are now ready to prove \cref{thm:hardness}.

\begin{proof}[Proof of \cref{thm:hardness}]
    Assuming $m=O(n)$ in the set cover instance, we present an AP reduction from \textsc{Set Cover} to both \BMR{} and \BSR{}. 
    
    \subparagraph*{\textsc{BMR}.} To produce a set cover solution, we take an improved approximate solution for \BMR{}, and output the family of sets whose corresponding versions are materialized. Since none of the $b_j$'s is stored, they have to be retrieved from some $a_i$. Moreover, under the constraint $\calR = 1$, they have to be a 1-hop neighbor of some $a_i$, meaning the materialized $a_i$ covers all of the elements in the set cover instance. 
    
    Finally, we prove that the approximation factor is preserved: for large $N$, the improved solution has objective value $\approx N\abs{\{ i: a_i \text{ materialized} \} }.$ If $n= O(m)$, then an $\alpha(\abs{V})$-approximation for \MMR{} provides a $(\alpha(n)+O(1))$-approximation for set cover. Hence we can not have $\alpha(\abs{V}) = c \ln n$ for $c<1$ unless $\NP\subseteq \DTIME(n^{O(\log\log n)})$~\cite{setCoverFeige1998}.

\begin{figure}[h!]
  \centering
  \includegraphics[width=.7\columnwidth]{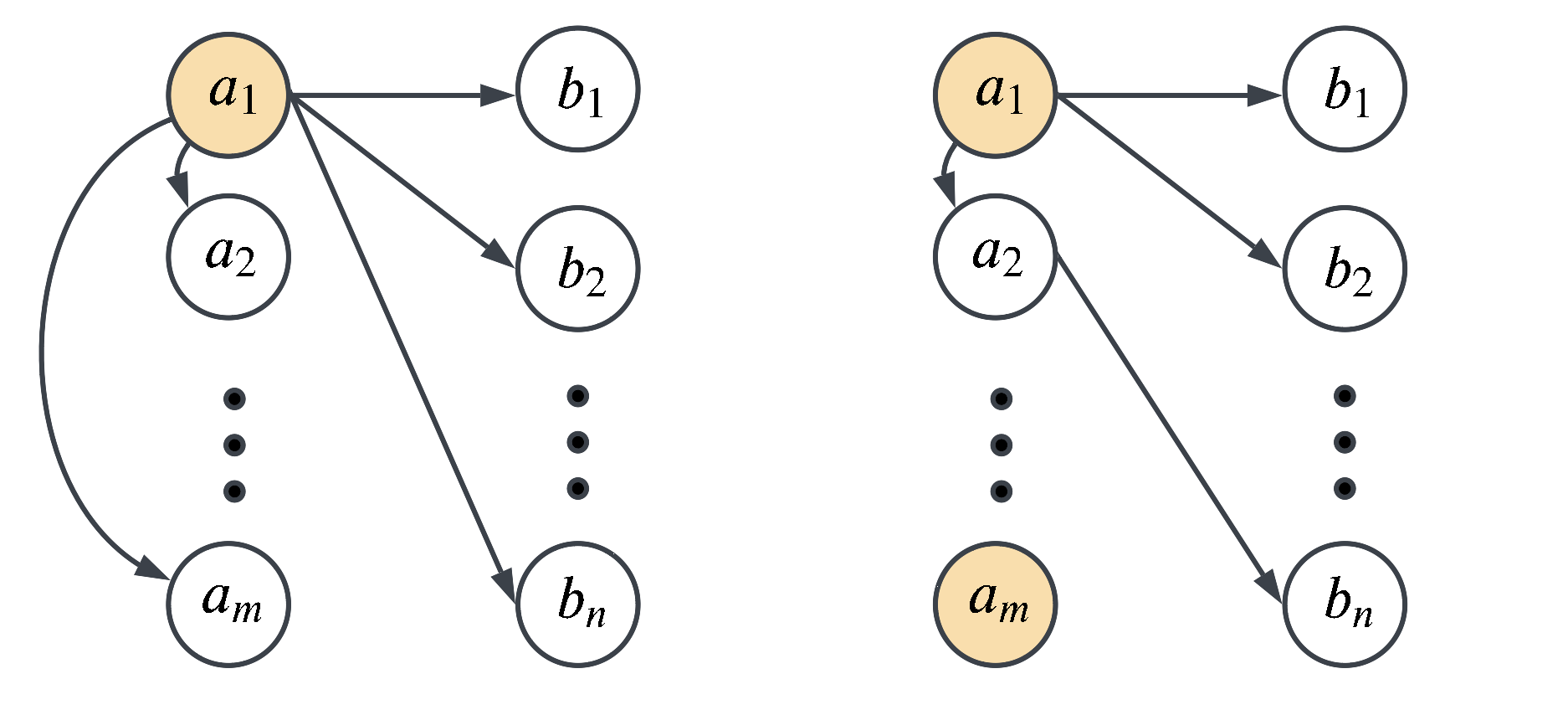}
    \caption{The \BSR{} case in proof of \cref{thm:hardness}. The solution on the right has one version ($b_2$) of retrieval cost 2, hence it must materialize an additional version $a_m$ to satisfy the total retrieval constraint. }
     \label{fig:hardness7}
\end{figure}

\subparagraph*{\textsc{BSR}.}
Assume for the moment that we know $m_{\OPT}$, then we can set total retrieval constraint to be $\calR=m-m_{\OPT}+n$, and work with an improved approximate solution. This choice of $\calR$ is made so that an optimal solution must materialize the set versions corresponding to a minimum set cover. All other nodes must be retrieved via a single hop. 

By \cref{lem:postProcessSolution2}, we assume all element versions are retrieved from a (not necessarily materialized) set version that covers it. If $m = O(n)$, an $\alpha(\abs{V})$-approximation of \BMR{} materializes $m_{\ALG}\leq (\alpha(n) + O(1))m_{\OPT}$ nodes. 

Note that, by materializing additional nodes, we are allowing a set $B$ of $b_j$'s to have retrieval cost $\geq 2$. Let $H$ denote the set of ``hopped sets'' $A_i$, which are not materialized yet are necessary to retrieve some $b_j$ through the delta $(a_i,b_j)$. By analyzing the total retrieval cost, we can bound $\abs{H}$ by:
$$\abs{H} \leq \abs{B} \leq m_{\ALG}-m_{\OPT}$$

Specifically, each additional $b_j\in B$ increases retrieval cost by at least $1$ compared to the optimal configuration; yet each of the $m_{\ALG}-m_{\OPT}$ additionally materialized set versions only decreases total retrieval cost by 1. It follows that the family of sets
$$S = \{A_i: a_i\text{ materialized }\}\cup H$$

is a $\Big (2\alpha(n) - O(1)\Big )$-approximation solution for the corresponding \textsc{Set Cover} instance. $S$ is feasible because all of the $b_j$'s are retrieved through some $(a_i,b_j)$, where $A_i\in S$; on the other hand, the size of both sets on the right hand side are at most $(\alpha(n)+O(1)) m_{\OPT}$, hence the approximation factor holds. 
Thus, any $\alpha(\abs{V})=c\ln n$ for any $c<0.5$ will result in a \textsc{Set Cover} approximation factor of $2c \cdot \ln(n)$. 

\subparagraph*{} We finish the proof by noting that, without knowing $m_{\OPT}$ in advance, we can run the above procedure for each possible guess of the value $m_{\OPT}$, and obtaining a feasible set cover each iteration. The desired approximation factor is still preserved by outputting the minimum set cover solution over the guesses. 
\end{proof} 

\bob{Not adding stuffs related to non-uniform demand.}

\subsection{Hardness on Arborescences}

We show that MSR and BSR are NP-hard on arborescence instances. This essentially shows that our FPTAS algorithm for MSR in \cref{subsec:tree-dp} is the best we can do in polynomial time. 
\begin{theorem}\label{thm:hardness-arborescence}
On arborescence inputs, \textsc{MinSum Retrieval} and \textsc{BoundedSum Retrieval} are \textsc{NP}-hard even when we assume single weight function and triangle inequality. 
\end{theorem}

In order to prove the theorem above, we rely on the following reduction which connects two problems together.

\begin{lemma}\label{lem:prob5SolvesProb3}
If there exists poly-time algorithm $\mathcal{A}$ that solves \textsc{BoundedSum Retrieval} (resp. \textsc{BoundedMax Retrieval}) on some set of input instances, then there exists a poly-time algorithm solving \textsc{MinSum Retrieval} (resp. \textsc{MinMax Retrieval}) on the same set of input instances. 
\end{lemma}
\begin{proof}
Suppose we want to solve a \MSR{} (resp. \MMR{}) instance with storage constraint $\calS$. We can use $\calA$ as a subroutine and conduct binary search for the minimum retrieval constraint $\calR^*$ under which \BSR{} (resp. \BMR{}) has optimal objective at most $\calS$. Thus, $\calR^*$ is an optimal solution for our problem at hand. 

To see that the binary search takes poly$(n)$ steps, we note that the search space for the target retrieval constraint is bounded by $n^2r_{max}$ for \BSR{} and $nr_{max}$ for \BMR{}, where $r_{max}=\max_{e\in E}r_e$. 
\end{proof}

Now we show the proof for \cref{thm:hardness-arborescence}.
\begin{proof}[Proof of~\cref{thm:hardness-arborescence}]
Assuming \cref{lem:prob5SolvesProb3}, it suffices to show the \textsc{NP}-hardness of \MSR{} on these inputs. 

Consider an instance of \textsc{Subset Sum} problem with values $a_1,\ldots,a_n$ and target $T$. This problem can be reduced to MSR on an $n$-ary arborescence of depth one. Let the root version be $v_0$ and its children $v_1,\ldots,v_n$. The materialization cost of $v_i$ is set to be $a_i+1$ for $i\in [n]$, while that of $v_0$ is some $N$ large enough so that the generalized triangle inequality holds. For each $i\in[n]$, we can set both retrieval and storage costs of edge $(v_0,v_i)$ to be $1$. 

Consider MSR on this graph with storage constraint $\calS=N+n+T$. From an optimal solution, we can construct set $A=\{i\in[n]: v_i\text{ materialized}\}$, an optimal solution for the above \textsc{Subset Sum} instance. 
\end{proof}

\section{Exact Algorithm for MMR and BMR on Bidirectional Trees}\label{Sec:arbor}
 \label{sec:MMRBMR}
As discussed in \cref{sec:intro}, we can use an algorithm for BMR to solve for MMR via binary search. Hence, it suffices to focus on \BMR{}, namely, when we are given maximal retrieval constraint $\calR$ and want to minimize storage cost. 

\setlength{\textfloatsep}{10pt}
\begin{algorithm}
\caption{\label{alg:DPforProb6} \textsf{DP-BMR}}
    \SetAlgoVlined
    \SetKwInOut{Input}{Input}
    \SetKwInOut{Output}{Output}
    \Input{Tree $T$ and the max retrieval constraint $\calR$}
    Orient $T$ arbitrarily. Sort $V$ in reverse topological order\;
    $\DP[v][u]\gets \infty$ for all $v,u\in V$\;
    \For{$v$ in $V$}{
        \For{$u$ in $V$ such that $R(u,v)\leq\calR$}{
            \eIf{$u = v$}{
                $\DP[v][u]\gets s_{v}$\;
            }{
                $\DP[v][u]\gets s_{p_v^u,v}$, where $p_v^u$ is the node preceding $v$ on the path from $u$ to $v$\;
            }
            \For{$w$ that is a child of $v$}{
                \eIf{$w$ in the path from $u$ to $v$}{
                $\DP[v][u]\gets \DP[v][u] + \DP[w][u]$\;
                }{
                $\DP[v][u] \gets \DP[v][u] + \min\{\OPT[w], \DP[w][u]\}$\;
                }
            }
            $\OPT[v] \gets \min\{\DP[v][w]: w\in V(T_{[v]})\}$\;
        }
    }
    \Return $\OPT[v_{root}]$\;
\end{algorithm}

Let $T=(V, E)$ be a bidirectional tree. This is a digraph with two directed edges $(u,v),(v,u) \in E$ corresponding to each edge $\{u,v\}\in E_0$, on some underlying undirected tree $(V,E_0)$. Let $\calR$ be the maximum retrieval cost constraint. 
We can pick any vertex $v_{0}$ as root, and orient the tree such that $v_0$ has no parent, while all other nodes have exactly one parent.

For each $v \in V$, let $T_{[v]}$ denote the subtree of $T$ rooted at $v$. If $v$ is retrieved from materialized $u$,
we use $p^u_v$ to denote the parent of $v$ on the unique $u-v$ path to retrieve $v$. 
We write $p^v_v = v$. 
We now describe a dynamic programming (DP) algorithm \textsf{DP-BMR} that solves \BMR{} exactly on $T$.

\subparagraph*{DP variables.} For $u,v\in V$, let $\DP[v][u]$ be the minimum storage cost of a \textit{partial solution} on $T_{[v]}$, which satisfies the following: all descendants of $v$ are retrieved from some node in $T_{[v]}$, while $v$ is retrieved from some materialized version $u$, which is \emph{potentially outside the subtree $T_{[v]}$}. See \cref{fig:DP4} for an illustration. 

Importantly, when calculating the storage cost for $DP[v][u]$, if $u$ is not a part of $T_{[v]}$, the incident edge $(p^u_v,v)$ is involved in the calculation, while other edges in the $u-v$ path, or the cost to materialize $u$, are not involved. 

\subparagraph*{Base case.} We iterate from the leaves up. Let $R(u,v)$ denote the retrieval cost of the $u-v$ path. For a leaf $v$, we set $DP[v][v] = s_{v}$, and $DP[v][u] = s_{(p^u_v,v)}$ for all $u\neq v$ with $R(u,v)\leq \calR$. Here, $p^u_v$ is just the parent of $v$ in the tree structure. All choices of $u,v$ such that $R(u,v) > \calR$ are infeasible, and we therefore set $\DP[v][u]=\infty$ in these cases. 

\subparagraph*{Recurrence.} 
For convenience, we define helper variable 
$OPT[v]$ to be the minimum storage cost on the subproblem $T_{[v]}$, such that \emph{$v$ is either materialized or retrieved from one of its materialized descendants}.\footnote{In other words, the case where $v$ is retrieved from $u$ outside of $T_{[v]}$, or case 3 in \cref{fig:DP4}, is not considered in this helper variable.} Formally, 

$$\OPT[v] = \min\{\DP[v][w]: w\in V(T_{[v]})\}.$$

For recurrence on $\DP[v][u]$ such that $R(u,v) \leq \calR$, there are three possible cases of the relationship between $v$ and $u$ (see \cref{fig:DP4}). In each case, we outline what we add to $\DP[v][u]$.

\begin{figure}\centering
\includegraphics[width=0.8\linewidth]{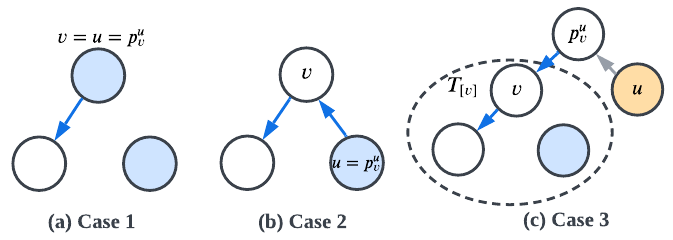}
\caption{The 3 cases of \textsf{DP-BMR}, where $u=v,u\in V(T_{[v]}),$ and $u\not\in V(T_{[v]})$ respectively. The blue nodes and edges are stored in the partial solution. }
\label{fig:DP4}
\end{figure}

\begin{bracketenumerate}
\item If $u = v$, we materialize $v$. Each child $w$ of $v$ can be either materialized, retrieved from their materialized descendants, or retrieved from the materialized $v$. Note that the storage cost on $T_{[w]}$ is exactly $\min\{\OPT[w],\DP[w][v]\}$, which we will add to the total value of $\DP[v][v]$. 

\item  If $u \in V(T_{[v]})\backslash \{v\}$, we would store the edge $(p^u_v,v)$. Note that $p^u_v$ is a child of $v$ and hence is also retrieved from the materialized $u$, so we must add $\DP[p^u_v][u]$. We then add $\min\{\OPT[w],\DP[w][u]\}$ for all other children $w$ of $v$. 

\item If $u\not\in V(T_{[v]})$, we would store the edge $(p^u_v,v)$, where $p^u_v$ is now the parent of $v$ in the tree structure. We then add $\min\{\OPT[w],\DP[w][u]\}$ for all children as before. 
\end{bracketenumerate}

\subparagraph*{Output.} We output $\OPT[v_{root}]$, the storage cost of the optimal solution. To output the configuration achieving this optimum, we can use the standard procedure where we store the configuration in each $\DP$ variable. 

\begin{theorem}
    \textsc{BoundedMax Retrieval} is solvable on bidirectional tree instances in $O(n^2)$ time. 
\end{theorem}

\begin{proof}

The time complexity follows from the observation that each calculation of $\DP[v][u]$ in the recurrence takes $O(\deg(v))$ time, and $\sum_u\sum_v \deg(v) = \sum_u O(n) = O(n^2)$. 

\bob{Rewritten this part.}
The optimality of this DP can be shown inductively from leaves up. On each leaf $v$, the optimal storage costs of the trivial subproblems are indeed $\DP[v][v] = s_v$, and $\DP[v][u] = s_{(p^u_v, v)}$ for all $u \neq v$ such that $R(u,v) \leq \calR$. 

Inductively, suppose node $v$ has children $w_1,\ldots, w_k$ on which the DP values are correctly calculated. To calculate the optimal storage cost $\DP[v][u]$ where $R(u,v) \leq \calR$, we consider $\DP[v][u]$ as the sum of the following three items:
\begin{bracketenumerate}
    \item The storage cost $s_{(p^u_v, v)}$, or $s_v$ if $u=v$. We add this cost directly since it is not part of the DP values of any child of $v$.
    \item The value $\min\{\DP[w_i][u], \OPT[u]\}$ for all child $w_i$ such that $u \not \in V(T_{[w_i]})$. This is becuase the minimum storage cost on subproblem $T_{[w_i]}$ is exactly $\min\{\DP[w_i][u], \OPT[u]\}$ if $v$ is not retrieved from any descendents of $w_i$. 
    \item $\DP[w_i][u]$ for child $w_i$ whose subtree $T_{[w_i]}$ contains $u$. This is because if so, for $v$ to be retrieved from $u$, $w_i$ must also be retrieved from $u$. Thus, we add $\DP[w_i][u]$ to $\DP[v][u]$, for this particular $w_i$. 
\end{bracketenumerate}

We also note that the partial solution on $T_{[w_i]}$ is completely independent from the partial solution on $T_{[w_j]}$, for all $i \neq j$. This allows us to directly sum over the individual optimal costs. 

The resulting $\DP[v][u]$ is also feasible: retrieving $v$ from materialized $u$ is feasible since $R(u,v) \leq \calR$, and any infeasible solution on $T_{[w_i]}$ is not considered due to its infinite $\DP$ value. 
\end{proof}

We note that by binary-searching the constraint value $\calS$, this algorithm also solves \textsc{MinMax Retrieval} on trees.

\section{FPTAS for \MSR{} via Dynamic Programming}
\label{sec:fptas}
In this section we work on \textsc{MinSum Retrieval} and present a fully polynomial time approximation scheme (FPTAS) on digraphs whose \textit{underlying undirected graph} has bounded treewidth. Similar techniques can be applied to \MMR{}, but we will focus on \MSR{} due to space constraints.


We start by describing a dynamic programming (DP) algorithm on trees in \cref{subsec:tree-dp}. In \cref{subsec:treewidth-def}, we define all notations necessary for the latter subsection. Finally, in \cref{subsec:DPBoundedTW}, we show how to extend our DP to the bounded treewidth graphs. 

\subsection{Warm-up: Bidirectional Trees}
\label{subsec:tree-dp}
As a warm-up to the more general algorithm, we present an FPTAS for bidirectional tree instances of \MSR{} via DP. This algorithm also inspired a practical heuristic \textsf{DP-MSR}, presented in \cref{subsubsec:Experiments/Algorithms/DP-heuristic}. 

Again, we assume the tree has a designated root $v_{root}$ and a parent-child hierarchy. We further assume that the tree is binary, via the standard trick of vertex splitting and adding edges of zero weight if necessary. (See \cref{appendix:WLOG-binary-tree} for details.)

\begin{figure}[h]
  \centering
  \includegraphics[width=0.9\columnwidth]{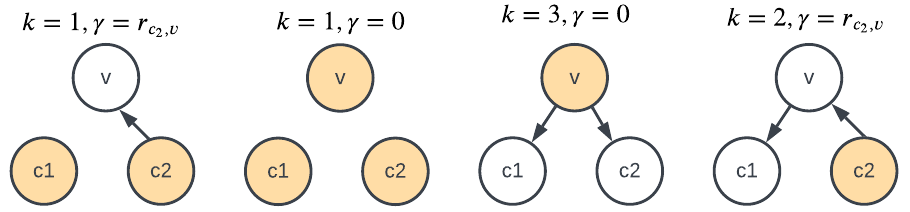}
 \caption{An illustration of DP variables in \cref{subsec:tree-dp}}
 \label{fig:DP3variabls}
\end{figure}

\subparagraph*{DP variables.} We define $\DP[v][k][\gamma][\rho]$ to be the minimum \textit{storage} cost for the subproblem with constraints $v,k,\gamma,\rho$ such that (with examples illustrated in \cref{fig:DP3variabls})
\begin{enumerate}
    \item \textit{Root for subproblem} $v \in V$ is a vertex on the tree; in each iteration, we consider the subtree rooted at $v$. 
    \item \textit{Dependency number} $k \in \mathbb{N}$ is the number of versions retrieved from $v$ (including $v$ itself) in the subproblem solution. This is useful when calculating the extra retrieval cost incurred by retrieving $v$ from its parent. 
    \item \textit{Root retrieval} $\gamma \in \mathbb{N}$ represents the cost of retrieving the subtree root $v$, if it is retrieved from a materialized descendant. This is useful when calculating the extra retrieval cost incurred by retrieving the parent of $v$ from $v$. Note that the root retrieval cost will be discretized, as specified later. 
    \item \textit{Total retrieval} $\rho \in \mathbb{N}$ represents the total retrieval cost of the subsolution. Similar to $\gamma$, $\rho$ will also be discretized.
\end{enumerate}

\begin{figure}[htp]
  \centering
  \includegraphics[width=0.9\columnwidth]{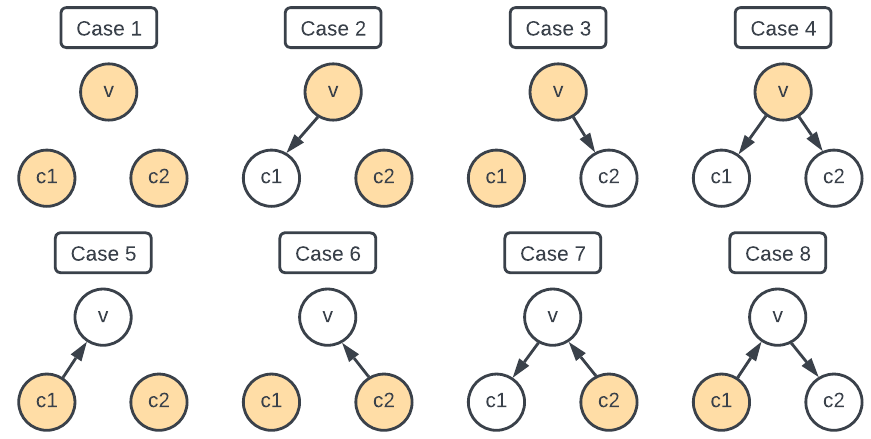}
 \caption{Eight types of connections on a binary tree. A node is colored if it is either materialized or retrieved from a node \textit{outside} the chart. Otherwise, an uncolored node is retrieved from another node as illustrated with the arrows.}
     \label{fig:8-connections}
\end{figure}

\subparagraph*{Discretizing retrieval costs.}\label{bullet-point:discretization} Let $r_{\mathrm{max}} = \max_{e\in E}\{r_e\}$. The possible total retrieval cost $\rho$ is within range $\{0,1,\ldots, n^2r_{\mathrm{max}}\}$. To make the DP tractable, we partition this range further and define \textit{approximated retrieval cost} $r'_{u,v}$ for edge $(u,v)\in E$ as follows:
$$
r'_{u,v}=\lceil \frac{r_{u,v}}{l} \rceil \quad \text{where }l = \frac{n^2r_{\mathrm{max}}}{t(\epsilon)},\, t(\epsilon)= \frac{n^4}{\epsilon},
$$
and $t(\epsilon)$ is the number of ``ticks'' we want to partition the retrieval range $[0, n^2 r_{\mathrm{max}}]$ into. The choice for $t(\epsilon)$ will be specified in the proof for \cref{thm:Tree-DP-FPTAS}. 
We will work with $r'$ in the rest of the subsection. However, by an abuse of notation, we still use $r$ for discretized retrieval for the ease of representation. 



\subparagraph*{Base case.} For each leaf $v$, we let $DP[v][1][0][0] = s_{v}$. 


\subparagraph*{Recurrence step.} On each iteration at node $v$, we consider every target configuration $\DP[v][k][\gamma][\rho]$ under each possible connection types as illustrated in \cref{fig:8-connections}. For each configuration, we go over all corresponding \textit{compatible} partial solutions on $T_{[c_1]}$ and $T_{[c_2]}$. 

\subparagraph*{} The recurrence relation for all cases is given in \cref{appendix:All-8-cases-tree-DP}. Here, we select representative cases and explain the details of calculation below:

\subsubsection{\textbf{Dealing with Dependency}}
When we decide to retrieve any child from $v$, as in case 4 of \cref{fig:8-connections}, the children $c_1,c_2$ along with all their dependencies now become dependencies of $v$. The minimum storage cost in case 4 (given $v,k,\gamma=0,\rho$) is: 
\begin{align} 
    S_4 &= s_v + s_{v,c_1} + s_{v,c_2} 
    - s_{c_1} - s_{c_2} \label{eq:dependencyline2}
    \\&+ \min_{\substack{\rho_1 + \rho_2 = \rho\\k_1 + k_2 = k-1}} \Big\{ 
     DP[c_1][k_1][0][\rho_1 - k_1 r_{v,c_1}] \tag{2a}\label{eq:dependencyline4}
    \\& + DP[c_2][k_2][0][\rho_2 - k_2 r_{v,c_2}]\Big\}\tag{2b}\label{eq:dependencyline5}
\end{align}

In \cref{eq:dependencyline4}, $v$ is required to have dependency number $k$ and root retrieval $0$. For each $k_1+k_2=k-1$, we must go through subproblems where $c_1$ has dependency number $k_1$ and $c_2$ has that of $k_2$.

Also in \cref{eq:dependencyline4}, the choice of $\rho_1, \rho_2$ determines how we are allocating retrieval costs budget $\rho$ to $c_1$ and $c_2$ respectively. Specifically in \cref{eq:dependencyline4} and \cref{eq:dependencyline5}, the total retrieval cost allocated to subproblem on $T_{[c_1]}$ is $\rho_1 - k_1\cdot r_{v,c_1}$ since an extra $k_1\cdot r_{v,c_1}$ cost is incurred by the edge $(v,c_1)$, as it is used $k_1$ times by all versions depending on $c_1$. Similar applies to the subproblem on $T_{[c_2]}$.


Next, we highlight the idea of ``invisible'' dependency here: in case 2 on $T_{[v]}$, the diffs $(v,c_1)$ and $(v,c_2)$ was not available in any previous recurrence, since $v$ has just been introduced. Therefore, the compatible solution for the subproblems on $T_{[c_1]}$ and $T_{[c_2]}$ have to materialize nodes $c_1$ and $c_2$ to ensure they can be retrieved. 
This explains the $-s_{c_1}-s_{c_2}$ terms in \cref{eq:dependencyline2}, since these costs are no longer present. 

When generalizing the DP onto graphs with bounded treewidth, similarly, restriction of a global solution does not always result in a feasible partial solution due to the existence of dependencies invisible to the subproblems. We will resolve them using similar ideas.


\subsubsection{\textbf{Dealing with Retrieval}}

In contrast with dependencies, this refers to the case where $v$ is retrieved from one of its children. We take case 5 as an example: given $v,k=0,\gamma,\rho$,
\begin{align*}
S_5 &= s_{c_1,v}
\\&+\min_{\rho_1 \le \rho}\Big\{\min_{k_1}\{DP[c_1][k_1][\gamma - r_{c_1,v}][\rho_1 - \gamma]\}
\\&+ \min_{k_2,\gamma'}\{ DP[c_2][k_2][\gamma'][\rho -\rho_1] \}\Big\}
\end{align*}

We allocate the retrieval cost similar to case 2. We will care less about the dependency number, over which we will take minimum. The retrieval cost for $c_1$ now has to be $\gamma - r_{c_1,v}$ since $v$ is retrieved from $c_1$. Note importantly that now we are counting the retrieval cost for $v$ in $\rho_1$, and so the retrieval cost budget for $T_{[c_1]}$ is now $\rho_1 - \gamma$. 

Similarly, we take minimum on all other unused parameters to get the best storage for case 5.
\subsubsection{\textbf{Combining the ideas}}

We take case 8 as an example where both retrieval and dependencies are involved. In case 8, $v$ is retrieved from child $c_1$ (retrieval), and child $c_2$ is retrieved from $v$ (dependency). Given $v,k,\gamma,\rho$, we claim that: 
\begin{align*}
    S_8 &= s_{c_1,v} + s_{v,c_2} - s_{c_2}\\
    &+ \min_{\rho_1+\rho_2 = \rho} \Big\{ \min_{k'}\{DP[c_1][k'][\gamma - r_{c_1,v}][\rho_1 - \gamma]\}\\
    & + DP[c_2][k-1][0][\rho_2 - (k-1)\cdot (r_2+\gamma)]\Big\}
\end{align*}
Note that the $c_1$ side is identical to that for case 5. In combining both dependency and retrieval cases, there is slight adjustment in the dependency side: since $v$ now might also depend on nodes further down $c_1$ side, the total extra retrieval cost created by adding edge $(v,c_2)$ becomes $(k-1)\cdot (r_2+\gamma)$ instead of $(k-1)\cdot (r_2)$.

\subparagraph*{Output.} Finally, with storage constraint $\mathcal{S}$ and root of the tree $v_{root}$, we output the configuration that outputs the minimum $\rho$ which achieves the following $$\exists k \leq n, \gamma \in \mathbb{N}\quad \text{ s.t. }\quad
DP[v_{root}][k][\gamma][\rho] \leq \mathcal{S}
$$

We shall formally state and prove the FPTAS result below. 

\begin{lemma}\label{lem:TreeFPTAS}
The DP algorithm outputs a configuration with total retrieval cost at most $\OPT + \epsilon r_{max}$ in poly$(n,1/\epsilon)$ time. 
\end{lemma}
\begin{proof}\label{Tree-DP-FPTAS-proof}
By setting $t(\epsilon) = \frac{n^4}{\epsilon}$, we have $l = \frac{n^2r_{max}}{t(\epsilon)}= \frac{\epsilon r_{max}}{n^2}$.
Note that we can get an approximation of the original retrieval costs by multiplying each $r_e'$ with $l$. This creates an estimation error of at most $l$ on each edge.
Note further that in the optimal solution, at most $n^2$ edges are materialized, so if $\rho^*$ is the minimal discretized total retrieval cost, we have
$$
    \text{total retrieval of output} \leq l\rho^*\leq \OPT + n^2l \leq \OPT + \epsilon r_{max}.
$$
\end{proof}

Now we prove the main theorem of this subsection:
\begin{theorem}\label{thm:Tree-DP-FPTAS}
For all $\epsilon>0$, there is a $(1+\epsilon)$-approximation algorithm for \textsc{MinSum Retrieval} on bidirectional trees that runs in poly$(n,\frac{1}{\epsilon})$ time. 
\end{theorem}

\begin{proof}
Given parameter $\epsilon$, we can use the DP algorithm as a black box and iterate the following for up to $n$ times:
\begin{bracketenumerate}
    \item Run the DP for the given $\epsilon$ on the current graph. Record the output. 

    \item Let $(u,v)$ be the most retrieval cost-heavy edge. We now set $r_{(u,v)}=0$ and $s_{(u,v)}=s_v$. If the new graph is infeasible for the given storage constraint, or if all edges have already been modified, exit the loop. 
\end{bracketenumerate}

At the end, we output the best out of all recorded outputs. This improves the previous bound when $r_{max} > \OPT$: at some point we will eventually have $r_{max}\leq \OPT$, which means the output configuration, if mapped back to the original input, is a feasible $(1+\epsilon)$-approximation. 
\end{proof}


\subsection{Treewidth-related Definitions}
\label{subsec:treewidth-def}
We now consider a more general class of version graphs: any $G=(V,E)$ whose \textit{underlying undirected graph}\footnote{As before, this means that $(u,v),(v,u)\in E$ for each undirected edge $\{u,v\}\in E_0$ in $G_0$. \bob{Added.}} $G_0$ has treewidth bounded by some constant $k$. 

\begin{definition}[Tree Decomposition~\cite{TreeWidthDefinition}]
A tree decomposition of undirected $G_0=(V_0,E_0)$ is a tree $T=(V_T,E_T)$, where each $z\in V_T$ is associated with a subset (``bag'') $S_z$ of $V_0$.
The bags must satisfy the following conditions:
\begin{romanenumerate}
    \item $\bigcup_{z\in V_T}S_z = V_0$;
    \item For each $v\in V_0$, the bags containing $v$ induce a connected subtree of $T$;
    \item For each $(u,v)\in E_0$, there exists $z\in V_T$ such that $S_z$ contains both $u$ and $v$. 
\end{romanenumerate}
The \emph{width} of a tree decomposition $T = (V_T,E_T)$ is $\max_{z\in V_T} \abs{S_z} - 1$. The \emph{treewidth} of $G_0$ is the minimum width over all tree decompositions of $G_0$. 
\end{definition}


It follows that undirected forests have treewidth 1. 
We further note that there is also a notion of directed treewidth \cite{DirectedTWDef}, but it is not suitable for our purpose.

We will WLOG assume a special kind of decomposition:


\begin{definition}[Nice Tree Decomposition~\cite{NiceTreeDecomp}]
A nice tree decomposition is a tree decomposition with a designated root, where each node $z$ is one of the following types:
\begin{enumerate}
    \item A \textbf{leaf}, which has no children and whose bag has size 1;
    \item A \textbf{forget node}, which has one children $c$, and $S_z\subset S_c$ and $\abs{S_c}=\abs{S_z}+1$. 

    \item An \textbf{introduce node}, which has one children $c$, and $S_z\supset S_c$ and $\abs{S_c}+1=\abs{S_z}$. 
    
    \item A \textbf{join}, which has children $c_1,c_2$, and $S_z=S_{c_1}=S_{c_2}$. 
\end{enumerate}
\end{definition}
Given a bound $k$ on the treewidth, there are multiple algorithms for calculating a tree decomposition of width $k$~\cite{exactTreewidth,exactTreewidth2,exactTreewidth3}, or an approximation of $k$~\cite{approxTreewidth1,approxTreewidth2,approxTreewidth3,approxTreewidth4}. 

For our case, the algorithm by Bodlaender~\cite{exactTreewidth2} can be used to compute a tree decomposition in time $2^{O(k^3)} \cdot O(n)$, which is linear if the treewidth $k$ is constant. Given a tree decomposition, we can in $O(\abs{V_0})$ time find a nice tree decomposition of the same width with $O(k\abs{V_0})$ nodes~\cite{NiceTreeDecomp}.

\subsection{Generalized Dynamic Programming}\label{subsec:DPBoundedTW}

Here we outline the DP for \MSR{} on graphs whose underlying undirected graph $G_0$ has treewidth at most $k$. 

\subsubsection{\textbf{DP States}} Similar to the warm-up, we will do the DP bottom-up on each $z\in V_T$ in the nice tree decomposition $T$. 

Before proceeding, let us define some additional notations. 
For any bag $z \in V_T$, let $T_{[z]}$ be the induced subtree of $T$ rooted at $z$. 
We define $V_{[z]}= \bigcup_{z' \in V(T_{z})} S_{z'}$ be the set of vertices in the bags of $T_{[z]}$, including $S_z$. Following that, $G_{[z]}$ is the induced subgraph of $G$ by vertices $V_{[z]}$.

We now define the \textit{DP states}. At a high level, each state describes some number of \textit{partial solutions} on the subgraph induced by $V_{[z]}$, $G_{[z]}$. When building a complete solution on $G$ from the partial solutions, the state variables should give us \textit{all} the information we need. 

Each DP state on $z\in V_T$ consists of a tuple of functions $$\calT_z = (\Par_z,\Dep_z,\Rec_z,\Anc_z)$$ and a non-negative integer $\rho_z$:
\begin{enumerate}
    \item \textit{Parent function} $\Par_z:S_z\mapsto V_{[z]}$ describing the partial solution on $G_{[z]}$, restricted on $S_z$. If $\Par_z(v)\neq v$ then $v$ will be retrieved through the edge $(\Par_z(v),v)$. If $\Par_z(v)=v$ then $v$ will be materialized. 
    
    \item \textit{Dependency function} $\Dep_z:S_z\mapsto [n]$. Similar to the dependency parameter in the warm-up, $\Dep_z(v)$ counts the number of nodes in $V_{[z]}$ retrieved through $v$. 

    \item \textit{Retrieval cost function} $\Rec_z: S_z\mapsto \{0,\ldots,nr_{\max}\}$. Similar to the root retrieval parameter in the warm-up, $\Rec_z(v)$ denotes the retrieval cost of version $v$ in the partial solution on $G_{[z]}$. 

    \item \textit{Ancestor function} $\Anc_z: S_z\mapsto 2^{S_z}$. 
    If $u\in\Anc_z(v)$, then $u$ is retrieved in 
    order to retrieve $v$ in this partial solution, 
    i.e., $v$ is dependent on $u$. We need this extra information to avoid directed cycles. 

    \item $\rho_z$, the total retrieval cost of the subproblem according to the partial solution. Similar to its counterpart in the warm-up, all retrieval costs would be discretized by the same technique that makes the approximation an FPTAS. 
\end{enumerate}

A feasible state on $z\in V_T$ is a pair $(\calT_z,\rho_z)$ which correctly describes some partial solution on $G_{[z]}$ whose retrieval cost is exactly $\rho_z$. Each state is further associated with a storage value $\sigma(\calT_z,\rho_z)\in \mathbb{Z}^+$, indicating the minimum storage needed to achieve the state $(\calT_z,\rho_z)$ on $G_{[z]}$. 


We are now ready to describe how to compute the states.

\subsubsection{\textbf{Recurrence on Leaves}} For each leaf $z\in V_T$, the only feasible partial solution is to materialize the only vertex $v$ in the leaf bag. We can easily calculate its state and storage cost. 

\subsubsection{\textbf{Recurrence on Forget Nodes}}
This is also easy: for a forget node $z$ with child $c$, we have $G_{[z]}=G_{[c]}$, and hence the states on $z$ are simply the restrictions of states on $c$. 

\subsubsection{\textbf{Recurrence on Introduce Nodes}}
At introduce node $z$ with child $c$, we have $S_z = S_c \cup \{v_0\}$ for some ``introduced'' $v_0$. Each feasible state $(\calT_z,\rho_z)$ on $z$ must correspond to some state $(\calT_c, \rho_c)$ on $c$, which we can calculate as follows:

We first initialize $\calT_c$ to be the respective functions in $\calT_z$ restricted on $S_c$. For instance, $\Par_c=\Par_z\mid_{S_c}$, the restriction of $\Par_z$ on domain $S_c$. 

If $v_0$ is retrieved through $u\in S_c$ according to $\calT_z$ ($\Par_z(v_0)=u$), then we remove the dependencies related to $v_0$ and the retrieval cost incurred on edge $(u,v_0)$. Specifically:
\begin{bracketenumerate}
    \item Decrease the value of $\Dep_c$ by 1 on all vertices in $\Anc_z(u)$.

    \item Decrease $\rho_c$ by $\Dep_z(v_0)\cdot \Rec_z(v_0)$. 

    \item Remove $\Anc_z(u)$ from the ancestor functions of all descendants of $z$. 
\end{bracketenumerate}

If $v_0$ has some child $w$ according to $\calT_z$ (namely, $\Par_z(w)=v_0$), then we reverse the \textit{uprooting} process in the warm-up, such that vertex $w$, which was not a root in $\calT_z$, is now a root in $\calT_c$. Specifically:

\begin{bracketenumerate}
    \item Let $\Par_c(w)=w$. 
    
    \item Remove $v_0$ from the ancestor function of $w$ and all its descendants.
    
    \item Decrease the retrieval cost function of $w$ and its descendants by $\Rec_z(w)$.
    
    \item Decrease $\rho_c$ by $\Rec_z(w)\times \Dep_z(w)$. 
\end{bracketenumerate}

Since $v$ could have multiple children, the last procedure is potentially repeated multiple times. 

\begin{algorithm}
    \caption{\label{alg:Compatibility} \textsc{Compatibility}}
    \SetKwInOut{Input}{Input}
    \Input{ $S_z,\calT_z,\calT_a,\calT_b$}

    \tcc{\textsc{External-Retrieval} returns the ``true restrictions'' of the $\Par$, $\Anc$, and $\Rec$ functions.}
    $\calT'_a,\calT'_b \gets $\textsc{External-Retrieval}$(S_z,\calT_z)$\;
    \If{$\calT'_a$ disagree with $\calT_a$ or $\calT'_b$ disagree with $\calT_b$ on functions $\Par,\Anc,$ or $\Rec$}{
        \Return \textbf{False};
    }
    \tcc{For each $v\in S_z$, \textsc{External_Dependency} returns the dependency of $v$ that are outside of $S_z$.}
    $\ExDep_z\gets$ \textsc{External_Dependency} $(S_z,\calT_z)$\;
    $\ExDep_a\gets$ \textsc{External_Dependency} $(S_z,\calT_a)$\;
    $\ExDep_b\gets$ \textsc{External_Dependency} $(S_z,\calT_b)$\;

        \If{$\ExDep_z \neq \ExDep_a + \ExDep_b$} {
            \Return \textbf{False}\;
        }
        \Return \textbf{True}\;
\end{algorithm}

\subsubsection{\textbf{Recurrence on Joins}} Suppose we are at a join $z$ with children $a,b$, where $S_z=S_a = S_b$. On a high level, for each state $(\calT_z,\rho_z)$ on $G_{[z]}$, we want to find all pairs of states $(\calT_a,\rho_a)$ and $(\calT_b,\rho_b)$ such that the partial solutions they describe can combine into a partial solution on $G_{[z]}$, as described by $(\calT_z,\rho_z)$. 

\subparagraph*{Compatibility. } The algorithm \textsc{Compatibility} (\cref{alg:Compatibility})
decides whether $\calT_a,\calT_b$ are indeed how $\calT_z$ looks like when restricted to $G_{[a]}$ and $G_{[b]}$ respectively. If the algorithm returns true, we proceed to calculate the correct value of $\rho_a+\rho_b$ based on this particular restriction.

\begin{figure}[h]
  \centering
  \includegraphics[width=0.8\columnwidth]{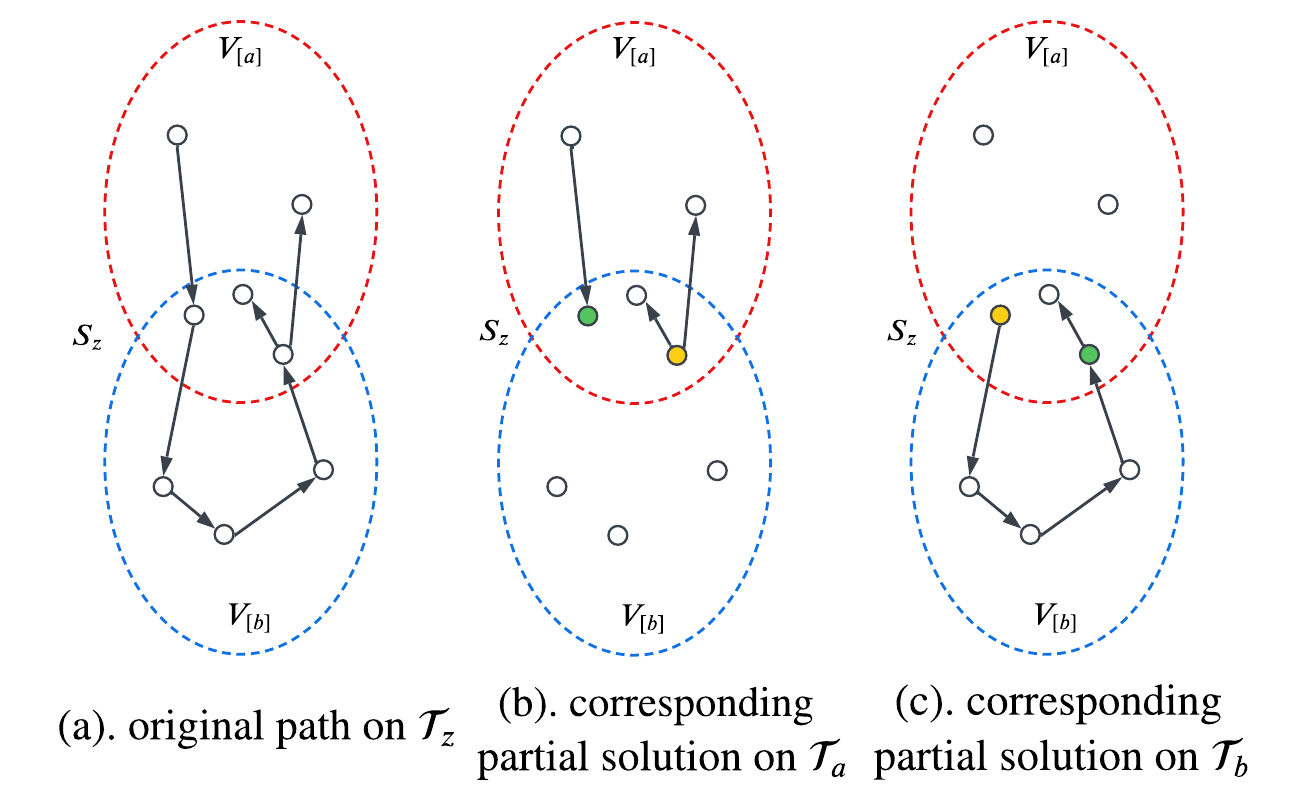}
 \caption{Illustration for compatibility. Figures (b) and (c) show a pair of compatible configurations on $\calT_a$ and $\calT_b$ with the configuration on $\calT_z$ in (a). The configurations of yellow nodes and green nodes are analyzed in \textsc{External-Retrieval} and \textsc{External-Dependency} respectively.}
     \label{fig:BTWCompatible}
\end{figure}

\subparagraph*{Resolving external retrieval.} \textsc{Compatibility} first deals with the vertices that are retrieved from outside $S_z$. For example, each $v\in S_z$ retrieved from $V_{[a]}\setminus S_z$, like the yellow node in (c) of \cref{fig:BTWCompatible}, is instead materialized from $\calT_b$'s perspective. To check whether $\calT_a$ and $\calT_b$ resolve all such cases correctly, we define subroutine \textsc{External-Retrieval} 
(\cref{alg:UnUproot}) 
to loop through $S_z$ topologically and calculate the correct $\Par, \Rec, \Anc$ functions for both $\calT_a$ and $\calT_b$. 

\begin{algorithm}
    \caption{\label{alg:UnUproot} \textsc{External-Retrieval}}
    \SetKwInOut{Input}{Input}
    \Input{ $S_z, \calT_z$}
    Let $\calT'_a = \calT'_b = \calT_z$\;
    Sort $S_z$ in topological order according to $\Anc_z$\;
    \For{$v\in S_z$} {
        \tcc{Removing external ancestors from $a$ iteratively.}
        \If{$\Par_z(v)\in V_{[a]}\setminus S_z$} {
            $\Par'_b(v)=v$\;
            \For{$w\in S_z$ with $w\neq v$ and $v\in \Anc'_b(w)$} {
                $\Rec'_b(w) \minuseq \Rec'_b(v)$\;
                $\Anc'_b(w) \gets \Anc'_b(w)\setminus \Anc'_b(v)$\;
            }
            $\Rec'_b(v) \gets 0$\;
            $\Anc'_b(v) \gets \varnothing$\;
        }

        \tcc{Removing external ancestors from $b$ iteratively.}
        \If{$\Par_z(v)\in V_{[b]}\setminus S_z$} {
            $\Par'_a(v)=v$\;
            \For{$w\in S)z$ with $w\neq v$ and $v\in \Anc'_a(w)$} {
                $\Rec'_a(w) \minuseq \Rec'_a(v)$\;
                $\Anc'_a(w) \gets \Anc'_a(w)\setminus \Anc'_a(v)$\;
            }
            $\Rec'_a(v) \gets 0$\;
            $\Anc'_a(v) \gets \varnothing$\;
        }
    }
    \Return $\calT'_a,\calT'_b$\;
\end{algorithm}

\subparagraph*{Resolving external dependency.} The next step in \textsc{Compatibility} is to check whether the functions $\Dep_a,\Dep_b$ are compatible with $\Dep_z$. Specifically, nodes in $S_z$ could have \textit{external dependencies} in $V_{[a]}\setminus S_z$ and $V_{[b]}\setminus S_z$, an example being the green nodes in \cref{fig:BTWCompatible} and \cref{fig:BTWExternalDependency}. The specific definition of $\ExDep_a(v)$ is the number of descendants that $v$ have outside $S_z$, to whom $v$ is the \textit{closest} ancestor in $S_z$, according to $\calT_a$. To see an example, note that only four red nodes are counted towards $\ExDep_a(A)$ in \cref{fig:BTWExternalDependency}. The functions $\ExDep_b$ and $\ExDep_z$ are defined similarly according to $\calT_b$ and $\calT_z$. 

We need that $\ExDep_a(v)+\ExDep_b(v)=\ExDep_z(v)$ for all  $v\in S_z$ in order for $(\calT_a,\calT_b)$ to be compatible with $\calT_z$. To check this, we call \textsc{External-Dependency} (\cref{alg:ExtDep}) on $\calT_z,\calT_a,\calT_b$ as a subroutine of \textsc{Compatibility.} We note that this is similar to distributing the dependency number $k$ to the two children in case $4$ of \cref{fig:8-connections}. 

\begin{algorithm}
    \caption{\label{alg:ExtDep} \textsc{External-Dependency}}
    \SetKwInOut{Input}{Input}
    \Input {$S, \calT$}
    Sort $S$ in topological order according to $\Anc$\;
    \For{$v\in S$}{
        Let $\ExDep(v) = \Dep(v) - \sum\limits_{w\in S:\Par(w)=v} \Dep(w)$\;
    }
    \For{$v\in S$}{
        \If{$\Par(v)\not\in S$}{
            \For{$u\in\Anc(v)$ with $u\neq v$}{
                $\ExDep(u)\minuseq\ExDep(v)$
            }
        }
    }
    \Return $\ExDep$\;
\end{algorithm}

\begin{figure}[h!]
  \centering
  \includegraphics[width=0.75\columnwidth]{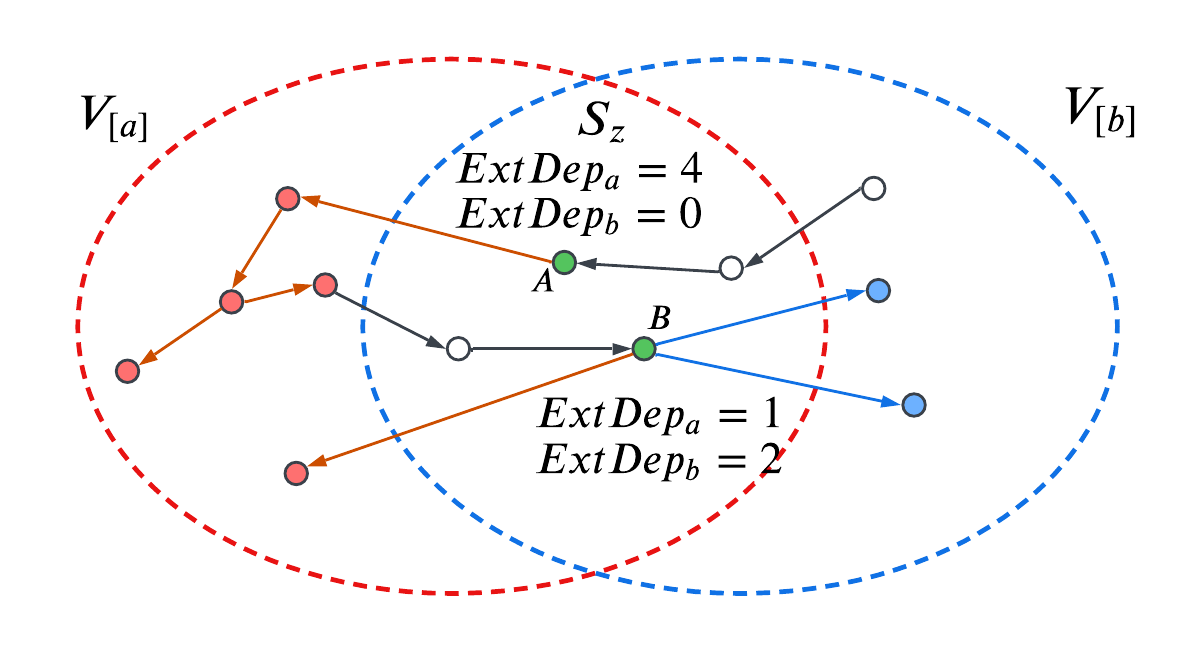}
 \caption{Illustration for external dependency. Green nodes $A$ and $B$ both have non-zero external dependency, as labeled in the figure. }
     \label{fig:BTWExternalDependency}
\end{figure}


\subparagraph*{Calculating $\rho$.} Given that $(\calT_a,\calT_b)$ are compatible with $\calT_z$, we want to find the objective, $\sigma(\calT_z,\rho_z)$, with the recurrence relation involving $\sigma(\calT_a,\rho_a)+\sigma(\calT_b,\rho_b)$ for suitable $\rho_a$ and $\rho_b$. However, we cannot simply take $\rho_a+\rho_b = \rho_z$ due to the complicated procedure of combining $\calT_a,\calT_b$ into $\calT_z$.
We thus implement \textsc{Distribute Retrieval} (\cref{alg:distributeRetrieval}) to calculate $\rho_\Delta$ such that $\rho_a+\rho_b  = \rho_z - \rho_\Delta$ and then iterate through all such $\rho_a$ and $\rho_b$. 

\begin{algorithm}
    \caption{\label{alg:distributeRetrieval} \textsc{Distribute Retrieval}}
    \SetKwInOut{Input}{Input}
    \Input {$S_z,\calT_z,\rho_z,S_a,S_b,\calT_a,\calT_b$}
    \tcc{We want $\rho_z = \rho_a+\rho_b + \rho_\Delta$:}
    $\rho_\Delta\gets 0$\; 
    \For {$v\in S_z$ such that $\Par_z(v)\neq v$} {
        \tcc{The number of times $r_{\Par_z(v),v}$ is counted towards $\rho_z$, minus the number of times it is counted towards $\rho_a$ and $\rho_b$:}
        $\Count\gets \Dep_z(v)$\; 
        \If {$\Par_a(v) = \Par_z(v)$} {
            $\Count \minuseq \Dep_a(v)$\;
        }
        \If {$\Par_b(v) = \Par_z(v)$} {
            $\Count \minuseq \Dep_b(v)$\;
        }
        \eIf {$\Par_z(v) \in S_z$} {
            \tcc{The edge $r_{\Par_z(v),v}$ is over/undercounted:}
            $\rho_\Delta \gets \rho_\Delta + \Count \cdot r_{\Par_z(v),v}$\; 
        }{ 
            \tcc{The entire $\Rec_z(v)$ is over/undercounted:} 
            $\rho_\Delta \gets \rho_\Delta + \Count \cdot \Rec_z(v)$\;
        }
    }
    \Return $\rho_\Delta$\;
\end{algorithm}


\subparagraph*{Recurrence relation.} Finally, we have all we need for the recurrence relation. For each feasible $(\calT_z,\rho_z)$, we take
$$
\sigma(\calT_z,\rho_z) = \min\left\{\sigma(\calT_a,\rho_a)+\sigma(\calT_b,\rho_b) - \textit{uproot} - \textit{overcount}\right\}
$$
where the minimum is taken over all $(\calT_a,\calT_b)$ that are compatible with $\calT_z$ and all $\rho_a+\rho_b=\rho_z-\rho_\Delta$, and where
$$
    \textit{uproot} = \sum_{v\in U_a} (s_{v}-s_{\Par_z(v),v}) + \sum_{v\in U_b} (s_{v}-s_{\Par_z(v),v}), \quad 
    \textit{overcount} = \sum_{v\in S_a\cap S_b} s_{\Par_z(v),v}.
$$

If $k$ is constant, then the recurrence relation takes poly$(n)$ time. This is because there are poly$(n)$ many possible choices of $\calT$ and $\rho$ on $a,b,z$,and it takes poly$(n)$ steps to check the compatibility of $(\calT_a,\calT_b)$ with $\calT_z$ and compute $\rho_\Delta$. 

\subparagraph*{Output.} The minimum retrieval cost of a global solution is just $\min\{\rho_z:\exists \calT_z,\sigma(\calT_z, \rho_z)\leq \calS\}$ over all feasible $(\calT_z,\rho_z)$, where $z$ is the designated root of the nice tree decomposition.

We conclude this section with the following theorem. 
\begin{theorem}\label{thm:BTW-DP-is-FPTAS}
    For a constant $k\geq 1$, on the set of graphs whose undelying undirected graph has treewidth at most $k$, \textsc{MinSum Retrieval} admits an FPTAS.
\end{theorem}

To see that our algorithm above is an FPTAS for \MSR{}, the proof is almost identical to the proof of \cref{thm:Tree-DP-FPTAS} (\cref{Tree-DP-FPTAS-proof}) once we note that the number of partial solutions on each $z$ is poly$(n)$. 

An FPTAS for \MMR{} arises from a similar procedure. When the objective becomes the maximum retrieval cost, we can use $\rho_z$ to represent the maximum retrieval cost in the partial solution. We then modify $\Dep_z(v)$ to represent the highest retrieval cost among all the nodes that are dependent on $v$. The recurrence relation is also changed accordingly. One can note that, like before, the new tuple $\calT_z$ contains all the information we need for a subsolution on $G_{[z]}$. 

The same algorithms extend to $(1,1+\epsilon)$ bi-criteria approximation algorithms for \BSR{} and \BMR{} naturally, as the objective and constraint are reversed.

\section{Heuristics on MSR and BMR}
In this section, we propose three new heuristics that are inspired by empirical observations and theoretical results. 

\subsection{\textsf{LMG-All}: Improvement over \textsf{LMG}}\label{sec:LMG-All}

We propose an improved version of \LMG{} (\cref{alg:LMG}), which we name \LMGA{}. (See \cref{alg:LMG-All} for pseudocode.) \LMGA{} enlarges the scope of the search on each greedy step. Instead of searching for the most efficient version to \textit{materialize}, we explore the payoff of \textit{modifying any single edge}:

\begin{enumerate}
    \item Find a configuration that minimizes total storage cost. 
    \item Let $\Par(v)$ be the current parent of $v$ on retrieval path. In addition to $V_{active}$, Define edge set $E_{active}$ to be the edges that (a) does not cause the configuration to exceed storage constraint $\calS$, and (b) does not form cycles, if $(u,v)\in E_{active}$ were to replace $(\Par(v),v)$ in the current configuration. If $V_{active}=E_{active}=\varnothing$, output the current configuration. 
    \item Calculate cost and benefit of each $v\in V_{active}$ and $e\in E_{active}$. Materialize or store the most cost-effective node or edge. Go to step 2 and repeat. 
\end{enumerate}

\begin{algorithm}
    \caption{\label{alg:LMG-All} \textsf{LMG-ALL}}
    \SetKwInOut{Input}{Input}
    \Input {Version graph $G$, storage constraint $\calS$}
    $G_{aux}\gets$ extended version graph with auxiliary root. 
    \tcc{See \cref{alg:LMG} for the construction of $G_{aux}$.}
    $T\gets$ minimum arborescence of $G_{aux}$ rooted at $v_{aux}$ with respect to weight function $s$\;
    Let $R(T)$ and $S(T)$ be the total retrieval and storage cost of $T$\;
    Let $P(v)$ be the parent of $v$ in $T$\;
    \While{$S(T) < \calS$}{
        $(\rho_{max},(u_{max},v_{max}))\gets (0,\varnothing)$\;
        \For{$e=(u,v)\in E$ where $u$ is not a descendant of $v$ in $T$}{
            $T_e = T\setminus (P(v),v)\cup\{e\}$\;
            \If {$R(T_e) > R(T)$} {
                \textbf{continue}\;
            }\eIf {$S(T_e) \leq S(T)$}{
                $\rho_e \gets \infty$;
            }{
                $\rho_{e}\gets (R(T)-R(T_e))/(s_e - s_{P(v),v})$;
            }

            \If{$\rho_e > \rho_{max}$}{
                $\rho_{max}\gets \rho_e$\;
                $(u_{max},v_{max})\gets e$\;
            }
        }
        \If{$\rho_{max} = 0$}{
            \Return $T$ \;
        }
        
        $T\gets T\setminus \{(P(v_{max}),v_{max})\} \cup \{(u_{max},v_{max})\}$\;
    }
    \Return $T$\;
\end{algorithm}

While \LMGA{} considers more edges than \LMG{}, it is not obvious that \LMGA{} always provides a better solution, due to its greedy nature.

\subsection{DP on Extracted Bidirectional Trees}
\label{subsubsec:Experiments/Algorithms/DP-heuristic}
We propose DP heuristics for both MSR and BMR, as inspired by algorithms in \cref{Sec:arbor,sec:fptas}.

To ensure a reasonable running time, we extract bidirectional trees\footnote{Recall this means a digraph whose underlying undirected graph is a tree, as in \cref{sec:fptas}} from input graphs and run the DP for treewidth 1 on the extracted graph, with the steps below:

\begin{bracketenumerate}
    \item Fix a node $v_{root}$ as the root. Calculate a minimum spanning arborescence $A$ of the graph $G$ rooted at $v_{root}$. We use the sum of retrieval and storage costs as weight. 
    \item Generate a bidirectional tree $G'$ from $A$. Namely, we have $(u,v),(v,u)\in E(G')$ for each edge $(u,v)\in E(A)$. 
    \item Run the proposed DP for MSR and BMR on directed trees (see \cref{subsec:tree-dp} and \cref{sec:MMRBMR}) with input $G'$.
\end{bracketenumerate}

In addition, we also implement the following modifications for \MSR{} to further speed up the algorithm: 
\begin{enumerate}
    \item Total \textit{storage} cost (instead of retrieval) is discretized and used as DP variable index, since it has a smaller range than retrieval cost.
    \item Geometric discretization is used instead of linear discretization, reducing the number of discretized ``ticks.''
    \item A pruning step is added, where the DP variable discards all subproblem solutions whose storage cost exceeds some threshold. This reduces redundant computations. 
\end{enumerate}

All three original features are necessary in the proof for our theoretical results, but in practice, the modified implementations show comparable results but significantly improves the running time. 
\section{Experiments for MSR and BMR}\label{sec:experiments}
In this section, we discuss the experimental setup and results for empirical validation of the algorithms' performance, as compared to
previous best-performing heuristic: \LMG{} for \MSR{}, and \MP{} for \BMR{}.\footnote{Our code can be found at https://anonymous.4open.science/r/Graph-Versioning-7343/README.md.}

In all figures, the vertical axis (objective and run time) is presented in \textit{logarithmic scale}. Run time is measured in \textit{milliseconds}. 

\subsection{Datasets and Construction of Graphs}
As in Bhattacherjee et al~\cite{versioningAmol}\bob{More justifications?}, we conduct experiment on real-world GitHub repositories of varying sizes as datasets. We construct version graphs as follows. Each commit corresponds to a node with its storage cost equal to its size in bytes. Between each pair of parent and child commits, we construct bidirectional edges. The storage and retrieval costs of the edges are calculated, in bytes, based on the actions (such as addition, deletion, and modification of files) required to change one version to the other in the direction of the edge. We use simple \texttt{diff} to calculate the deltas, hence the storage and retrieval costs are proportional to each other. Graphs generated this way are called ``\textbf{natural graphs}'' in the rest of the section. 

In addition, we also aim to test (1) the cases where the retrieval and storage costs of an edge can greatly differ from each other, and (2) the effect of tree-like shapes of graphs on the performance of algorithms. Therefore, we also conduct experiments on modified graphs in the following two ways:

\subparagraph*{Random compression.} 
We simulate compression of data by scaling storage cost with a rand
om factor between 0.3 and 1, and increasing the retrieval cost by 20\% (to simulate decompression). The resulting storage and retrieval costs are potentially very different. 

\begin{table}[t]
\begin{center}
    \footnotesize
    \begin{tabular}{|l||c|c|c|c|}
        \hline
        \textbf{Dataset} & \textbf{\#nodes} & \textbf{\#edges} & \textbf{avg. cost $s_v$} & \textbf{avg. cost $s_e$} \\
        \hline 
        \texttt{datasharing} & 29 & 74 & 7672 & 395 \\
        \hline 
        \texttt{styleguide} & 493 & 1250 & $1.4\times 10^6$ & 8659 \\
        \hline 
        \texttt{996.ICU} & 3189 & 9210 & $1.5\times 10^7$ & 337038 \\
        \hline
        \texttt{freeCodeCamp} & 31270 & 71534 & $2.5 \times 10^7$ & 14800 \\
        \hline
        \texttt{LeetCodeAnimation} & 246 & 628 & $1.7 \times 10^8$ & $1.2\times 10^7$ \\
        \hline 
        \texttt{LeetCode} (0.05) & 246 & 3032 & $1.7 \times 10^8$ & $1.0 \times 10^8$ \\
        \hline 
        \texttt{LeetCode} (0.2) & 246 & 11932 & $1.7 \times 10^8$ & $1.0 \times 10^8$ \\
        \hline 
        \texttt{LeetCode} (1) & 246 & 60270 & $1.7 \times 10^8$ & $1.0 \times 10^8$ \\
        \hline 
    \end{tabular}
\end{center}
\caption{Natural and ER graphs overview. }
\label{table:datasets}
\end{table}

\subparagraph*{ER construction.} Instead of the naturally constructing edges between each pair of parent and child commits, we construct the edges as in an Erdős-Rényi random graph: between each pair $(u,v)$ of versions, with probability $p$ both deltas $(u,v)$ and $(v,u)$ are constructed, and with probability $1-p$ neither are constructed. The resulting graphs are much less tree-like.\footnote{ER graphs have treewidth $\Theta(n)$ with high probability if the number of edges per vertex is greater than a small constant~\cite{ERTreewidth}.} We construct ER graphs from the repository LeetCode because it has a moderate size and is the least tree-like.\footnote{On LeetCode, the average unnatural delta is 10 times more costly than a natural delta. This ratio is around 100 for other repositories.} \bob{Added explanations for using LeetCode. }






\subsection{Results in MSR}

\begin{figure}[h]
  \centering
  \includegraphics[width=0.55\columnwidth]{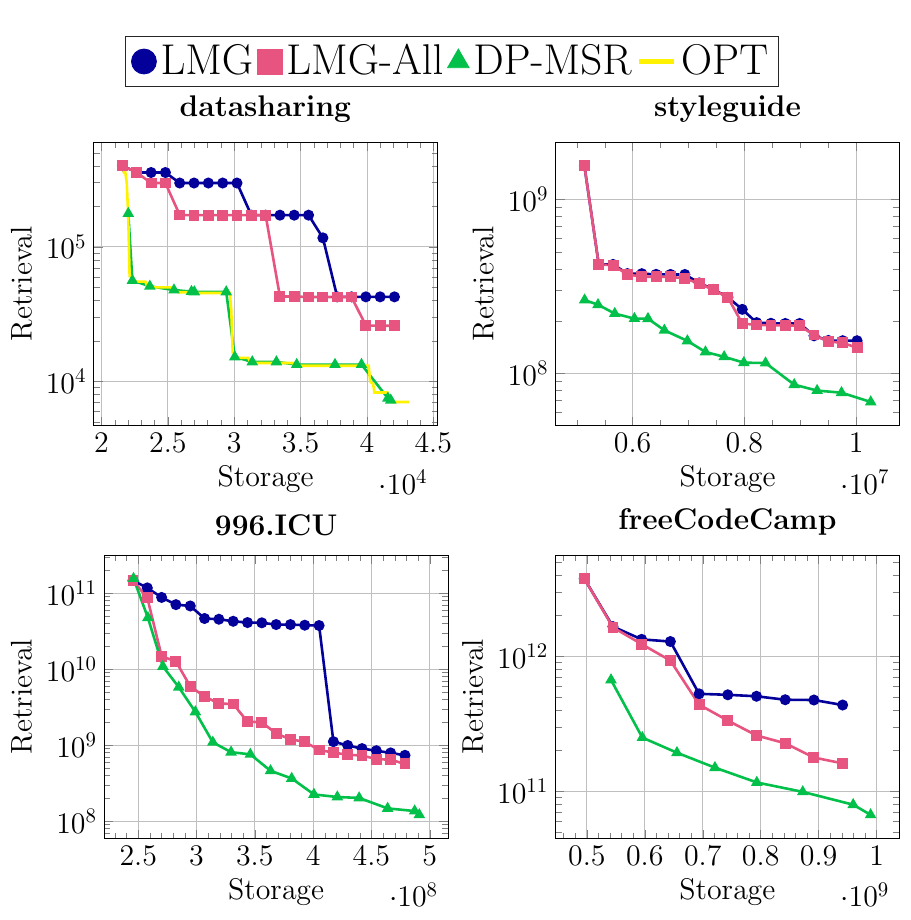}
 \caption{Performance of MSR algorithms on natural graphs. OPT is obtained by solving an integer linear program (ILP, see \cref{appendix:ILP}) using Gurobi~\cite{gurobi}. ILP takes too long to finish on all graphs except \texttt{datasharing}. }
 \label{fig:MSR-performance-comparison}
\end{figure}

\begin{figure}[h]
  \centering
  \includegraphics[width=0.55\columnwidth]{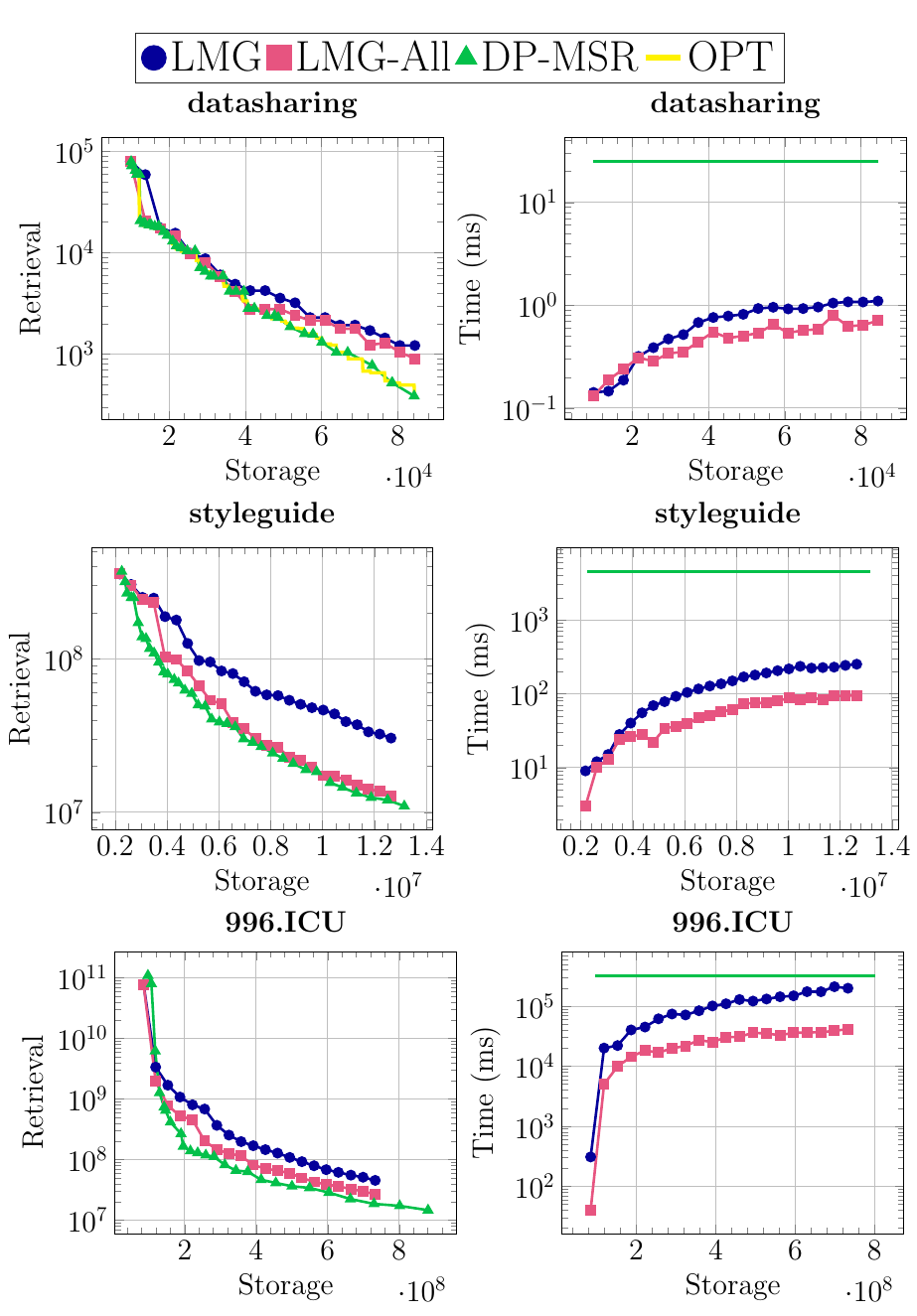}
 \caption{Performance and run time of MSR algorithms on compressed graphs.}
     \label{fig:MSR-performance-comparison-compressed}
\end{figure}

\begin{figure}[h]
  \centering
  \includegraphics[width=0.55\columnwidth]{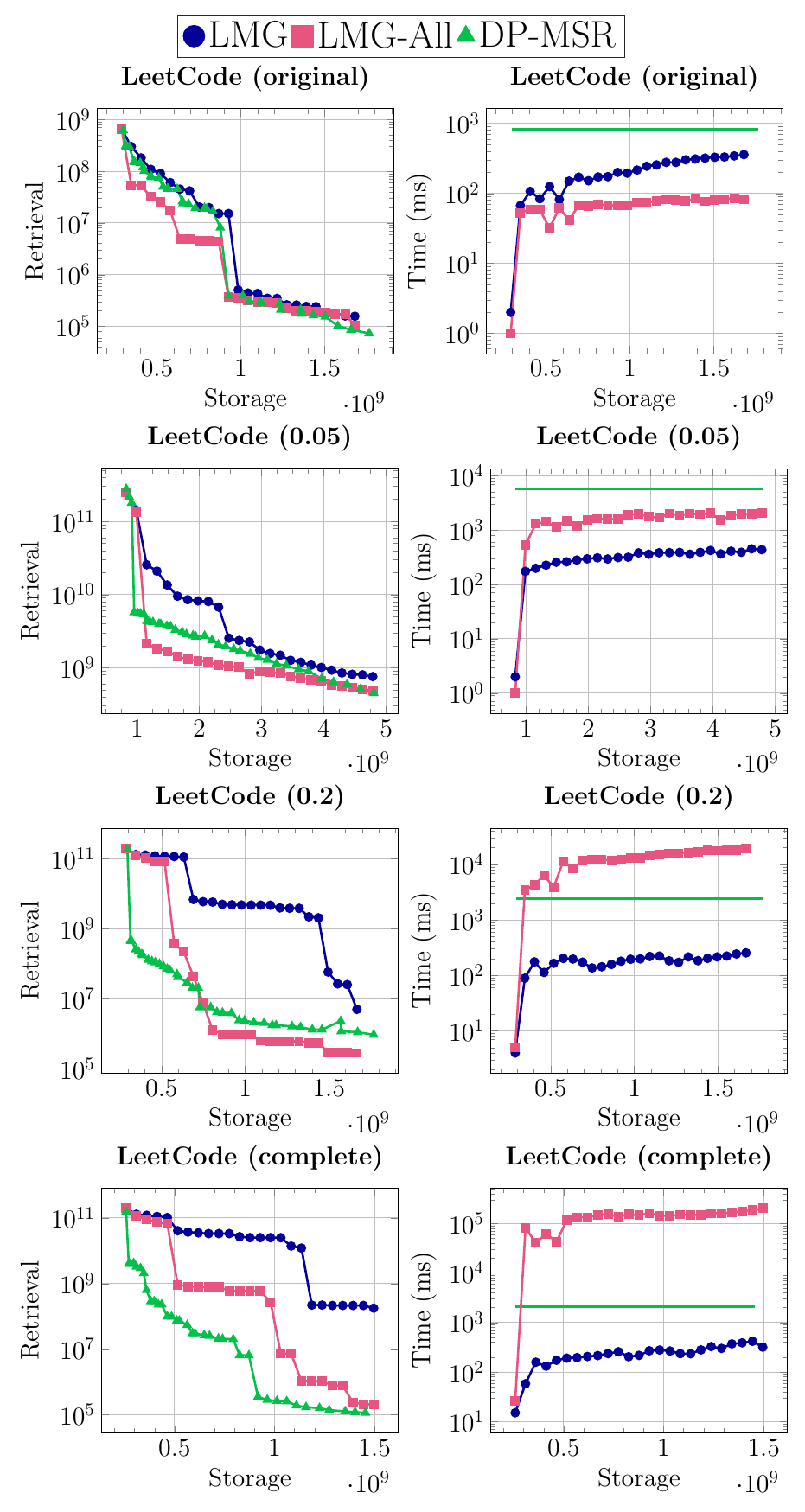}
 \caption{Performance and run time of MSR algorithms on compressed ER graphs.}
     \label{fig:MSR-performance-comparison-compressed-ER}
\end{figure}

\cref{fig:MSR-performance-comparison}, \cref{fig:MSR-performance-comparison-compressed}, and \cref{fig:MSR-performance-comparison-compressed-ER} demonstrate the performance of the three MSR algorithms on natural graphs, compressed natural graphs, and compressed ER graphs. The running times for the algorithms are shown in \cref{fig:MSR-performance-comparison-compressed} and \cref{fig:MSR-performance-comparison-compressed-ER}. Since run time for most non-ER graphs exhibit similar trends, many are omitted here due to space constraint. Also note that, since \DPMSR{} generates all data points in a single run, its running time is shown as a horizontal line over the full range for storage constraint.

We run \DPMSR{} with $\epsilon=0.05$ on most graphs, except $\epsilon = 0.1$ for \texttt{freeCodeCamp} (for the feasibility of run time). The pruning value for DP variables is at twice the minimum storage for uncompressed graphs, and ten times the minimum storage for randomly compressed graphs. 

\subparagraph*{Performance analysis.} 
On most graphs, \DPMSR{} outperforms \LMGA{}, which in turn outperforms \LMG{}. This is especially clear on natural version graphs, where \DPMSR{} solutions are near 1000 times better than \LMG{} solutions on \texttt{996.ICU}. 
in \cref{fig:MSR-performance-comparison}. On \texttt{datasharing}, \DPMSR{} almost perfectly matches the optimal solution (calculated via the ILP in \cref{appendix:ILP}) for all constraint ranges. 

On naturally constructed graphs (\cref{fig:MSR-performance-comparison}), \LMGA{} often has comparable performance with \LMG{} when storage constraint is low. This is possibly because both algorithms can only iterate a few times when the storage constraint is almost tight. \DPMSR, on the other hand, performs much better on natural graphs even for low storage constraint.

On graphs with random compression (\cref{fig:MSR-performance-comparison-compressed}), the dominance of DP in performance over the other two algorithms become less significant. This is anticipated because of the fact that DP only runs on a subgraph of the input graph. Intuitively, most of the information is already contained in a minimum spanning tree when storage and retrieval costs are proportional. Otherwise, the dropped edges may be useful. (They could have large retrieval but small storage, and vice versa. )

Finally, \LMG{}'s performance relative to our new algorithms is much worse on ER graphs. This may be due to the fact that \LMG{} cannot look at non-auxiliary edges once the minimum arborescence is initialized, and hence losing most of the information brought by the extra edges. (\cref{fig:MSR-performance-comparison-compressed-ER}). 

\subparagraph*{Run time analysis. }For all natural graphs, we observe that \LMGA{} uses no more time than \LMG{} (as shown in \cref{fig:MSR-performance-comparison-compressed}). Moreover, \LMGA{} is significantly quicker than \LMG{} on large natural graphs, which was unexpected considering that the two algorithms have almost identical structures in implementation. Possibly, this could be due to \LMG{} making bigger, more expensive changes on each iteration (materializing a node with many dependencies, for instance) as compared to \LMGA{}.


As expected, though, \LMGA{} takes much more time than the other two algorithms on denser ER graphs (\cref{fig:MSR-performance-comparison-compressed-ER}), due to the large number of edges. 

\DPMSR{} is often slower than \LMG{}, except when ran on the natural construction of large graphs (\cref{fig:MSR-performance-comparison-compressed}). However, unlike \LMG{} and \LMGA{}, the DP algorithm returns a whole spectrum of solutions at once, so it is difficult to make a direct comparison. We also note that the run time of DP heavily depends on the choice of $\epsilon$ and the storage pruning bound. Hence, the user can trade-off the run time with solution's qualities by parameterize the algorithm with coarser configurations.

\subsection{Results in BMR}

As compared to MSR algorithms, the performance and run time of our BMR algorithms are much more predictable and stable. They exhibit similar trends across different ways of graph construction as mentioned earlier in this section - including the non-tree-like ER graphs, surprisingly. 

\bob{Maybe we can remove this can just claim that the same trend occurs on every other dataset.}
Due to space limitation, we only present the results on natural graphs, as shown in \cref{fig:BMR-time-comparison}, to respectively illustrate their performance and run time. 

\subparagraph*{Performance analysis. }
For every graph we tested, \DPBMR{} outperforms MP on most of the retrieval constraint ranges. As the retrieval constraint increases, the gap between MP and \DPBMR{} solution also increases. We also observe that \DPBMR{} performs worse than MP when the retrieval constraint is at zero. This is because the bidirectional tree have fewer edges than the original graph. (Recall that the same behavior happened for \DPMSR{} on compressed graphs) 

We also note that, unlike MP, the objective value of \DPBMR{} solution monotonically decreases with respect to retrieval constraint. This is again expected since these are optimal solutions of the problem on the bidirectional tree.  

\subparagraph*{Run time analysis. }
For all graphs, the run times of \DPBMR{} and MP are comparable within a constant factor. This is true with varying graph shapes and construction methods in all our experiments, and representative data is exhibited in \cref{fig:BMR-time-comparison}. Unlike \LMG{} and \LMGA{}, their run times do not change much with varying constraint values. 

\begin{figure}[h!]
  \centering
  \includegraphics[width=.55\columnwidth]{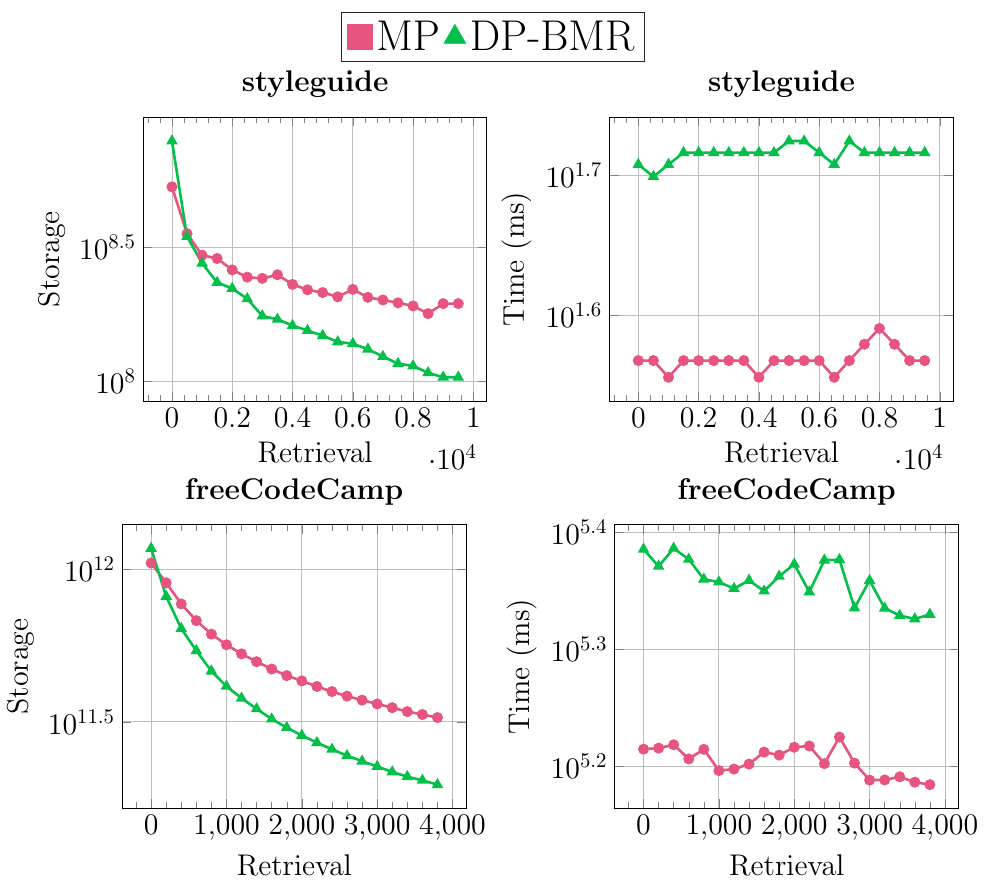}
 \caption{Performance and run time of BMR algorithms on natural version graphs. }
     \label{fig:BMR-time-comparison}
\end{figure}

\subparagraph*{Overall evaluation}
For MSR, we recommend always using one of \LMGA{} and \DPMSR{} in place of \LMG{} for practical use. On sparse graphs, \LMGA{} dominates \LMG{} both in performance and run time. \DPMSR{} can also provide a frontier of better solutions in a reasonable amount of time, regardless of the input. 

For BMR, \DPBMR{} usually outperforms MP, except when the retrieval constraint is close to zero. Therefore, we recommend using DP in most situations. 

\section{Conclusion}
In this paper, we developed fully polynomial time approximation algorithms for graphs with bounded treewidth. This often captures the typical manner in which edit operations are applied on versions. 
For practical use, we extracted the idea behind this approach as well as previous \LMG{} approach, and developed heuristics which significantly improved both the performance and run time in experiments, while potentially allowing for parallelization. 



\bibliographystyle{plainurl}
\bibliography{ref}

\begin{thebibliography}{10}

\bibitem{git}
{G}it.
\newblock https://github.com/git/git, 2005.
\newblock last accessed: 13-Oct-22.

\bibitem{pachyderm}
Pachyderm.
\newblock https://github.com/pachyderm/pachyderm, 2016.
\newblock last accessed: 13-Oct-22.

\bibitem{DVC}
{DVC}.
\newblock https://github.com/iterative/dvc, 2017.
\newblock last accessed: 13-Oct-22.

\bibitem{Dolt}
{Dolt}.
\newblock https://github.com/dolthub/dolt, 2019.
\newblock last accessed: 13-Oct-22.

\bibitem{TerminusDB}
{TerminusDB}.
\newblock https://github.com/terminusdb/terminusdb, 2019.
\newblock last accessed: 13-Oct-22.

\bibitem{LakeFS}
{LakeFS}.
\newblock https://github.com/treeverse/lakeFS, 2020.
\newblock last accessed: 13-Oct-22.

\bibitem{anwar2015taming}
Ali Anwar, Yue Cheng, Aayush Gupta, and Ali~R Butt.
\newblock Taming the cloud object storage with mos.
\newblock In {\em Proceedings of the 10th Parallel Data Storage Workshop}, pages 7--12, 2015.

\bibitem{anwar2016mos}
Ali Anwar, Yue Cheng, Aayush Gupta, and Ali~R Butt.
\newblock Mos: Workload-aware elasticity for cloud object stores.
\newblock In {\em Proceedings of the 25th ACM International Symposium on High-Performance Parallel and Distributed Computing}, pages 177--188, 2016.

\bibitem{unpublishedArcher}
Aaron Archer.
\newblock Inapproximability of the asymmetric facility location and k-median problems.
\newblock 2000.

\bibitem{k-center-approx}
Aaron Archer.
\newblock Two o(log*k)-approximation algorithms for the asymmetric k-center problem.
\newblock In Karen Aardal and Bert Gerards, editors, {\em Integer Programming and Combinatorial Optimization}, pages 1--14, Berlin, Heidelberg, 2001. Springer Berlin Heidelberg.

\bibitem{exactTreewidth2}
Stefan Arnborg, Derek~G. Corneil, and Andrzej Proskurowski.
\newblock Complexity of finding embeddings in a k-tree.
\newblock {\em SIAM Journal on Algebraic Discrete Methods}, 8(2):277--284, 1987.
\newblock \href {https://arxiv.org/abs/https://doi.org/10.1137/0608024} {\path{arXiv:https://doi.org/10.1137/0608024}}, \href {https://doi.org/10.1137/0608024} {\path{doi:10.1137/0608024}}.

\bibitem{approxTreewidth1}
Mahdi Belbasi and Martin F\"{u}rer.
\newblock Finding all leftmost separators of size $\leq k$.
\newblock In {\em Combinatorial Optimization and Applications: 15th International Conference, COCOA 2021, Tianjin, China, December 17–19, 2021, Proceedings}, page 273–287, Berlin, Heidelberg, 2021. Springer-Verlag.
\newblock \href {https://doi.org/10.1007/978-3-030-92681-6_23} {\path{doi:10.1007/978-3-030-92681-6_23}}.

\bibitem{TreeWidthDefinition}
Umberto Bertelè and Francesco Brioschi.
\newblock On non-serial dynamic programming.
\newblock {\em Journal of Combinatorial Theory, Series A}, 14(2):137--148, 1973.
\newblock \href {https://doi.org/10.1016/0097-3165(73)90016-2} {\path{doi:10.1016/0097-3165(73)90016-2}}.

\bibitem{bhardwaj14DataHub}
Anant~P. Bhardwaj, Souvik Bhattacherjee, Amit Chavan, Amol Deshpande, Aaron~J. Elmore, Samuel Madden, and Aditya~G. Parameswaran.
\newblock Datahub: Collaborative data science {\&} dataset version management at scale.
\newblock In {\em Seventh Biennial Conference on Innovative Data Systems Research, {CIDR} 2015, Asilomar, CA, USA, January 4-7, 2015, Online Proceedings}. www.cidrdb.org, 2015.
\newblock URL: \url{http://cidrdb.org/cidr2015/Papers/CIDR15\_Paper18.pdf}.

\bibitem{versioningAmol}
Souvik Bhattacherjee, Amit Chavan, Silu Huang, Amol Deshpande, and Aditya~G. Parameswaran.
\newblock Principles of dataset versioning: Exploring the recreation/storage tradeoff.
\newblock {\em Proc. {VLDB} Endow.}, 8(12):1346--1357, 2015.
\newblock URL: \url{http://www.vldb.org/pvldb/vol8/p1346-bhattacherjee.pdf}, \href {https://doi.org/10.14778/2824032.2824035} {\path{doi:10.14778/2824032.2824035}}.

\bibitem{exactTreewidth}
Hans~L. Bodlaender.
\newblock A linear time algorithm for finding tree-decompositions of small treewidth.
\newblock In {\em Proceedings of the Twenty-Fifth Annual ACM Symposium on Theory of Computing}, STOC '93, page 226–234, New York, NY, USA, 1993. Association for Computing Machinery.
\newblock \href {https://doi.org/10.1145/167088.167161} {\path{doi:10.1145/167088.167161}}.

\bibitem{NiceTreeDecomp}
Hans~L. Bodlaender.
\newblock A partial k-arboretum of graphs with bounded treewidth.
\newblock {\em Theoretical Computer Science}, 209(1):1--45, 1998.
\newblock \href {https://doi.org/10.1016/S0304-3975(97)00228-4} {\path{doi:10.1016/S0304-3975(97)00228-4}}.

\bibitem{bogatu20DDD}
Alex Bogatu, Alvaro A.~A. Fernandes, Norman~W. Paton, and Nikolaos Konstantinou.
\newblock {Dataset Discovery in Data Lakes}.
\newblock {\em 2020 IEEE 36th International Conference on Data Engineering (ICDE)}, 00:709--720, 2020.
\newblock \href {https://arxiv.org/abs/2011.10427} {\path{arXiv:2011.10427}}, \href {https://doi.org/10.1109/icde48307.2020.00067} {\path{doi:10.1109/icde48307.2020.00067}}.

\bibitem{brickley19GoogleDataSet}
Dan Brickley, Matthew Burgess, and Natasha Noy.
\newblock {Google Dataset Search: Building a search engine for datasets in an open Web ecosystem}.
\newblock {\em The World Wide Web Conference}, pages 1365--1375, 2019.
\newblock \href {https://doi.org/10.1145/3308558.3313685} {\path{doi:10.1145/3308558.3313685}}.

\bibitem{brown21DSDB}
Jackson Brown and Nicholas Weber.
\newblock {DSDB: An Open-Source System for Database Versioning \& Curation}.
\newblock {\em 2021 ACM/IEEE Joint Conference on Digital Libraries (JCDL)}, 00:299--307, 2021.
\newblock \href {https://doi.org/10.1109/jcdl52503.2021.00044} {\path{doi:10.1109/jcdl52503.2021.00044}}.

\bibitem{buneman00DataProvenance}
Peter Buneman, Sanjeev Khanna, and Wang-Chiew Tan.
\newblock {Data Provenance: Some Basic Issues}.
\newblock {\em Lecture Notes in Computer Science}, pages 87--93, 2000.
\newblock \href {https://doi.org/10.1007/3-540-44450-5\_6} {\path{doi:10.1007/3-540-44450-5\_6}}.

\bibitem{burn98inplace}
Randal~C. Burns and Darrell D.~E. Long.
\newblock {In-place reconstruction of delta compressed files}.
\newblock {\em Proceedings of the seventeenth annual ACM symposium on Principles of distributed computing - PODC '98}, pages 267--275, 1998.
\newblock \href {https://doi.org/10.1145/277697.277747} {\path{doi:10.1145/277697.277747}}.

\bibitem{chavan17Dex}
Amit Chavan and Amol Deshpande.
\newblock {DEX: Query Execution in a Delta-based Storage System}.
\newblock {\em Proceedings of the 2017 ACM International Conference on Management of Data}, pages 171--186, 2017.
\newblock \href {https://doi.org/10.1145/3035918.3064056} {\path{doi:10.1145/3035918.3064056}}.

\bibitem{cheng2015cast}
Yue Cheng, M~Safdar Iqbal, Aayush Gupta, and Ali~R Butt.
\newblock Cast: Tiering storage for data analytics in the cloud.
\newblock In {\em Proceedings of the 24th International Symposium on High-Performance Parallel and Distributed Computing}, pages 45--56, 2015.

\bibitem{chimani15NDesign}
Markus Chimani and Joachim Spoerhase.
\newblock {Network Design Problems with Bounded Distances via Shallow-Light Steiner Trees}.
\newblock In Ernst~W. Mayr and Nicolas Ollinger, editors, {\em 32nd International Symposium on Theoretical Aspects of Computer Science (STACS 2015)}, volume~30 of {\em Leibniz International Proceedings in Informatics (LIPIcs)}, pages 238--248, Dagstuhl, Germany, 2015. Schloss Dagstuhl--Leibniz-Zentrum fuer Informatik.
\newblock URL: \url{http://drops.dagstuhl.de/opus/volltexte/2015/4917}, \href {https://doi.org/10.4230/LIPIcs.STACS.2015.238} {\path{doi:10.4230/LIPIcs.STACS.2015.238}}.

\bibitem{k-center-hardness}
Julia Chuzhoy, Sudipto Guha, Eran Halperin, Sanjeev Khanna, Guy Kortsarz, Robert Krauthgamer, and Joseph~(Seffi) Naor.
\newblock Asymmetric k-center is log* n-hard to approximate.
\newblock {\em J. ACM}, 52(4):538–551, jul 2005.
\newblock \href {https://doi.org/10.1145/1082036.1082038} {\path{doi:10.1145/1082036.1082038}}.

\bibitem{AP-Reduction}
P.~Crescenzi.
\newblock A short guide to approximation preserving reductions.
\newblock In {\em Proceedings of Computational Complexity. Twelfth Annual IEEE Conference}, pages 262--273, 1997.
\newblock \href {https://doi.org/10.1109/CCC.1997.612321} {\path{doi:10.1109/CCC.1997.612321}}.

\bibitem{Derakhshan22material}
Behrouz Derakhshan, Alireza Rezaei~Mahdiraji, Zoi Kaoudi, Tilmann Rabl, and Volker Markl.
\newblock Materialization and reuse optimizations for production data science pipelines.
\newblock SIGMOD '22, page 1962–1976, New York, NY, USA, 2022. Association for Computing Machinery.
\newblock \href {https://doi.org/10.1145/3514221.3526186} {\path{doi:10.1145/3514221.3526186}}.

\bibitem{devarajan2020hcompress}
Hariharan Devarajan, Anthony Kougkas, Luke Logan, and Xian-He Sun.
\newblock Hcompress: Hierarchical data compression for multi-tiered storage environments.
\newblock In {\em 2020 IEEE IPDPS}, pages 557--566. IEEE, 2020.

\bibitem{devarajan2020hfetch}
Hariharan Devarajan, Anthony Kougkas, and Xian-He Sun.
\newblock Hfetch: Hierarchical data prefetching for scientific workflows in multi-tiered storage environments.
\newblock In {\em 2020 IEEE IPDPS}, pages 62--72. IEEE, 2020.

\bibitem{setCoverNPhard}
Irit Dinur and David Steurer.
\newblock Analytical approach to parallel repetition.
\newblock In {\em Proceedings of the Forty-Sixth Annual ACM Symposium on Theory of Computing}, STOC '14, page 624–633, New York, NY, USA, 2014. Association for Computing Machinery.
\newblock \href {https://doi.org/10.1145/2591796.2591884} {\path{doi:10.1145/2591796.2591884}}.

\bibitem{ERRADI2020110457}
Abdelkarim Erradi and Yaser Mansouri.
\newblock Online cost optimization algorithms for tiered cloud storage services.
\newblock {\em Journal of Systems and Software}, 160:110457, 2020.
\newblock \href {https://doi.org/10.1016/j.jss.2019.110457} {\path{doi:10.1016/j.jss.2019.110457}}.

\bibitem{setCoverFeige1998}
Uriel Feige.
\newblock A threshold of ln n for approximating set cover.
\newblock {\em J. ACM}, 45(4):634–652, jul 1998.
\newblock \href {https://doi.org/10.1145/285055.285059} {\path{doi:10.1145/285055.285059}}.

\bibitem{approxTreewidth4}
Uriel Feige, MohammadTaghi Hajiaghayi, and James~R. Lee.
\newblock Improved approximation algorithms for minimum-weight vertex separators.
\newblock In {\em Proceedings of the Thirty-Seventh Annual ACM Symposium on Theory of Computing}, STOC '05, page 563–572, New York, NY, USA, 2005. Association for Computing Machinery.
\newblock \href {https://doi.org/10.1145/1060590.1060674} {\path{doi:10.1145/1060590.1060674}}.

\bibitem{fernandez18Aurum}
Raul~Castro Fernandez, Ziawasch Abedjan, Famien Koko, Gina Yuan, Sam Madden, and Michael Stonebraker.
\newblock {Aurum: A Data Discovery System}.
\newblock {\em 2018 IEEE 34th International Conference on Data Engineering (ICDE)}, pages 1001--1012, 2018.
\newblock \href {https://doi.org/10.1109/icde.2018.00094} {\path{doi:10.1109/icde.2018.00094}}.

\bibitem{approxTreewidth2}
Fedor~V. Fomin, Daniel Lokshtanov, Saket Saurabh, Micha\L{} Pilipczuk, and Marcin Wrochna.
\newblock Fully polynomial-time parameterized computations for graphs and matrices of low treewidth.
\newblock {\em ACM Trans. Algorithms}, 14(3), jun 2018.
\newblock \href {https://doi.org/10.1145/3186898} {\path{doi:10.1145/3186898}}.

\bibitem{exactTreewidth3}
Fedor~V. Fomin, Ioan Todinca, and Yngve Villanger.
\newblock Large induced subgraphs via triangulations and cmso.
\newblock {\em SIAM Journal on Computing}, 44(1):54--87, 2015.
\newblock \href {https://arxiv.org/abs/https://doi.org/10.1137/140964801} {\path{arXiv:https://doi.org/10.1137/140964801}}, \href {https://doi.org/10.1137/140964801} {\path{doi:10.1137/140964801}}.

\bibitem{ERTreewidth}
Yong Gao.
\newblock Treewidth of erd{\H{o}}s--r{\'e}nyi random graphs, random intersection graphs, and scale-free random graphs.
\newblock {\em Discrete Applied Mathematics}, 160(4-5):566--578, 2012.

\bibitem{subsetSumFPTAS3}
GB~Gens and YV~Levner.
\newblock Approximate algorithms for certain universal problems in scheduling theory.
\newblock {\em Engineering Cybernetics}, 16(6):31--36, 1978.

\bibitem{subsetSumFPTAS4}
George Gens and Eugene Levner.
\newblock A fast approximation algorithm for the subset-sum problem.
\newblock {\em INFOR: Information Systems and Operational Research}, 32(3):143--148, 1994.

\bibitem{ghuge2022quasi}
Rohan Ghuge and Viswanath Nagarajan.
\newblock Quasi-polynomial algorithms for submodular tree orienteering and directed network design problems.
\newblock {\em Mathematics of Operations Research}, 47(2):1612--1630, 2022.

\bibitem{SymmetrickCenterHardness}
Teofilo~F Gonzalez.
\newblock Clustering to minimize the maximum intercluster distance.
\newblock {\em Theoretical computer science}, 38:293--306, 1985.

\bibitem{gurobi}
{Gurobi Optimization, LLC}.
\newblock {Gurobi Optimizer Reference Manual}, 2022.
\newblock URL: \url{https://www.gurobi.com}.

\bibitem{HaeuplerEmbedding}
Bernhard Haeupler, D.~Ellis Hershkowitz, and Goran Zuzic.
\newblock Tree embeddings for hop-constrained network design.
\newblock In {\em Proceedings of the 53rd Annual ACM SIGACT Symposium on Theory of Computing}, STOC 2021, page 356–369, New York, NY, USA, 2021. Association for Computing Machinery.
\newblock \href {https://doi.org/10.1145/3406325.3451053} {\path{doi:10.1145/3406325.3451053}}.

\bibitem{hajiaghayi2009approximating}
Mohammad~Taghi Hajiaghayi, Guy Kortsarz, and Mohammad~R Salavatipour.
\newblock Approximating buy-at-bulk and shallow-light k-steiner trees.
\newblock {\em Algorithmica}, 53(1):89--103, 2009.

\bibitem{huang20OrpheusDB}
Silu Huang, Liqi Xu, Jialin Liu, Aaron~J. Elmore, and Aditya Parameswaran.
\newblock {ORPHEUSDB: bolt-on versioning for relational databases (extended version)}.
\newblock {\em The VLDB Journal}, 29(1):509--538, 2020.
\newblock \href {https://doi.org/10.1007/s00778-019-00594-5} {\path{doi:10.1007/s00778-019-00594-5}}.

\bibitem{hunt98Delta}
James~J. Hunt, Kiem-Phong Vo, and Walter~F. Tichy.
\newblock {Delta algorithms: an empirical analysis}.
\newblock {\em ACM Transactions on Software Engineering and Methodology (TOSEM)}, 7(2):192--214, 1998.
\newblock \href {https://doi.org/10.1145/279310.279321} {\path{doi:10.1145/279310.279321}}.

\bibitem{subsetSumFPTAS1}
Oscar~H. Ibarra and Chul~E. Kim.
\newblock Fast approximation algorithms for the knapsack and sum of subset problems.
\newblock {\em J. ACM}, 22(4):463–468, oct 1975.
\newblock \href {https://doi.org/10.1145/321906.321909} {\path{doi:10.1145/321906.321909}}.

\bibitem{symmetric-k-median-hardness-correct}
Kamal Jain, Mohammad Mahdian, and Amin Saberi.
\newblock A new greedy approach for facility location problems.
\newblock In {\em Proceedings of the Thiry-Fourth Annual ACM Symposium on Theory of Computing}, STOC '02, page 731–740, New York, NY, USA, 2002. Association for Computing Machinery.
\newblock \href {https://doi.org/10.1145/509907.510012} {\path{doi:10.1145/509907.510012}}.

\bibitem{jayawardana19DFS}
Yasith Jayawardana and Sampath Jayarathna.
\newblock {DFS: A Dataset File System for Data Discovering Users}.
\newblock {\em 2019 ACM/IEEE Joint Conference on Digital Libraries (JCDL)}, 00:355--356, 2019.
\newblock \href {https://arxiv.org/abs/1905.13363} {\path{arXiv:1905.13363}}, \href {https://doi.org/10.1109/jcdl.2019.00068} {\path{doi:10.1109/jcdl.2019.00068}}.

\bibitem{DirectedTWDef}
Thor Johnson, Neil Robertson, {P. D.} Seymour, and Robin Thomas.
\newblock Directed tree-width.
\newblock {\em Journal of Combinatorial Theory. Series B}, 82(1):138--154, May 2001.
\newblock Funding Information: 1Partially supported by the NSF under Grant DMS-9701598. 2 Research partially supported by the DIMACS Center, Rutgers University, New Brunswick, NJ 08903. 3Partially supported by the NSF under Grant DMS-9401981. 4Partially supported by the ONR under Contact N00014-97-1-0512. 5Partially supported by the NSF under Grant DMS-9623031 and by the NSA under Contract MDA904-98-1-0517.
\newblock \href {https://doi.org/10.1006/jctb.2000.2031} {\path{doi:10.1006/jctb.2000.2031}}.

\bibitem{subsetSumFPTAS2}
Richard~M Karp.
\newblock The fast approximate solution of hard combinatorial problems.
\newblock In {\em Proc. 6th South-Eastern Conf. Combinatorics, Graph Theory and Computing (Florida Atlantic U. 1975)}, pages 15--31, 1975.

\bibitem{subsetSumFPTAS5}
Hans Kellerer, Renata Mansini, Ulrich Pferschy, and Maria~Grazia Speranza.
\newblock An efficient fully polynomial approximation scheme for the subset-sum problem.
\newblock {\em Journal of Computer and System Sciences}, 66(2):349--370, 2003.

\bibitem{RezaBuyAtBulk}
M.~Reza Khani and Mohammad~R. Salavatipour.
\newblock Improved approximations for buy-at-bulk and shallow-light k-steiner trees and (k,2)-subgraph.
\newblock {\em J. Comb. Optim.}, 31(2):669–685, feb 2016.
\newblock \href {https://doi.org/10.1007/s10878-014-9774-5} {\path{doi:10.1007/s10878-014-9774-5}}.

\bibitem{khuller93LAST}
Samir Khuller, Balaji Raghavachari, and Neal~E. Young.
\newblock Balancing minimum spanning trees and shortest-path trees.
\newblock {\em Algorithmica}, 14(4):305--321, 1995.
\newblock \href {https://doi.org/10.1007/BF01294129} {\path{doi:10.1007/BF01294129}}.

\bibitem{khurana12EffSnap}
Udayan Khurana and Amol Deshpande.
\newblock {Efficient Snapshot Retrieval over Historical Graph Data}.
\newblock {\em arXiv}, 2012.
\newblock Graph database systems --- stroing dynamic graphs so that a graph at a specific time can be queried. Vertices are marked with bits encoding information on which versions it belong to.
\newblock \href {https://arxiv.org/abs/1207.5777} {\path{arXiv:1207.5777}}, \href {https://doi.org/10.48550/arxiv.1207.5777} {\path{doi:10.48550/arxiv.1207.5777}}.

\bibitem{kinoshita2021cost}
Reika Kinoshita, Satoshi Imamura, Lukas Vogel, Satoshi Kazama, and Eiji Yoshida.
\newblock Cost-performance evaluation of heterogeneous tierless storage management in a public cloud.
\newblock In {\em 2021 Ninth International Symposium on Computing and Networking (CANDAR)}, pages 121--126. IEEE, 2021.

\bibitem{approxTreewidth3}
Tuukka Korhonen.
\newblock A single-exponential time 2-approximation algorithm for treewidth.
\newblock In {\em 2021 IEEE 62nd Annual Symposium on Foundations of Computer Science (FOCS)}, pages 184--192, 2022.
\newblock \href {https://doi.org/10.1109/FOCS52979.2021.00026} {\path{doi:10.1109/FOCS52979.2021.00026}}.

\bibitem{kortsarz1997approximating}
Guy Kortsarz and David Peleg.
\newblock Approximating shallow-light trees.
\newblock In {\em Proceedings of the eighth annual ACM-SIAM symposium on Discrete algorithms}, pages 103--110, 1997.

\bibitem{hermes}
Anthony Kougkas, Hariharan Devarajan, and Xian-He Sun.
\newblock Hermes: a heterogeneous-aware multi-tiered distributed i/o buffering system.
\newblock In {\em Proceedings of the 27th International Symposium on High-Performance Parallel and Distributed Computing}, pages 219--230, 2018.

\bibitem{lasch2022cost}
Robert Lasch, Thomas Legler, Norman May, Bernhard Scheirle, and Kai-Uwe Sattler.
\newblock Cost modelling for optimal data placement in heterogeneous main memory.
\newblock {\em Proceedings of the VLDB Endowment}, 15(11):2867--2880, 2022.

\bibitem{lasch2021workload}
Robert Lasch, Robert Schulze, Thomas Legler, and Kai-Uwe Sattler.
\newblock Workload-driven placement of column-store data structures on dram and nvm.
\newblock In {\em Proceedings of the 17th International Workshop on Data Management on New Hardware (DaMoN 2021)}, pages 1--8, 2021.

\bibitem{liu2019transfer}
Mingyu Liu, Li~Pan, and Shijun Liu.
\newblock To transfer or not: An online cost optimization algorithm for using two-tier storage-as-a-service clouds.
\newblock {\em IEEE Access}, 7:94263--94275, 2019.

\bibitem{liu2021keep}
Mingyu Liu, Li~Pan, and Shijun Liu.
\newblock Keep hot or go cold: A randomized online migration algorithm for cost optimization in staas clouds.
\newblock {\em IEEE Transactions on Network and Service Management}, 18(4):4563--4575, 2021.

\bibitem{macdonald2000file}
Josh MacDonald.
\newblock {\em File system support for delta compression}.
\newblock PhD thesis, Masters thesis. Department of Electrical Engineering and Computer Science, University of California at Berkley, 2000.

\bibitem{maddox16Decibel}
Michael Maddox, David Goehring, Aaron~J. Elmore, Samuel Madden, Aditya Parameswaran, and Amol Deshpande.
\newblock {Decibel: The Relational Dataset Branching System}.
\newblock {\em Proceedings of the VLDB Endowment. International Conference on Very Large Data Bases}, 9(9):624--635, 2016.
\newblock \href {https://doi.org/10.14778/2947618.2947619} {\path{doi:10.14778/2947618.2947619}}.

\bibitem{manne22CHEX}
Naga~Nithin Manne, Shilvi Satpati, Tanu Malik, Amitabha Bagchi, Ashish Gehani, and Amitabh Chaudhary.
\newblock {CHEX: multiversion replay with ordered checkpoints}.
\newblock {\em Proceedings of the VLDB Endowment}, 15(6):1297--1310, 2022.
\newblock \href {https://doi.org/10.14778/3514061.3514075} {\path{doi:10.14778/3514061.3514075}}.

\bibitem{marathe1998bicriteria}
Madhav~V Marathe, Ramamoorthi Ravi, Ravi Sundaram, SS~Ravi, Daniel~J Rosenkrantz, and Harry~B Hunt~III.
\newblock Bicriteria network design problems.
\newblock {\em Journal of algorithms}, 28(1):142--171, 1998.

\bibitem{koyel-ICDE}
Koyel Mukherjee, Raunak Shah, Shiv~K. Saini, Karanpreet Singh, Khushi ~, Harsh Kesarwani, Kavya Barnwal, and Ayush Chauhan.
\newblock Towards optimizing storage costs on the cloud.
\newblock {\em IEEE 39th International Conference on Data Engineering (ICDE) (To Appear)}, 2023.

\bibitem{nagel2006subversion}
William Nagel.
\newblock Subversion: not just for code anymore.
\newblock {\em Linux Journal}, 2006(143):10, 2006.

\bibitem{dataLake}
Fatemeh Nargesian, Erkang Zhu, Ren\'{e}e~J. Miller, Ken~Q. Pu, and Patricia~C. Arocena.
\newblock Data lake management: Challenges and opportunities.
\newblock {\em Proc. VLDB Endow.}, 12(12):1986–1989, aug 2019.
\newblock \href {https://doi.org/10.14778/3352063.3352116} {\path{doi:10.14778/3352063.3352116}}.

\bibitem{RRaviBroadcast}
R.~Ravi.
\newblock Rapid rumor ramification: approximating the minimum broadcast time.
\newblock In {\em Proceedings 35th Annual Symposium on Foundations of Computer Science}, pages 202--213, 1994.
\newblock \href {https://doi.org/10.1109/SFCS.1994.365693} {\path{doi:10.1109/SFCS.1994.365693}}.

\bibitem{databricks}
Paul Roome, Tao Feng, and Sachin Thakur.
\newblock Announcing the availability of data lineage with unity catalog.
\newblock https://www.databricks.com/blog/2022/06/08/announcing-the-availability-of-data-lineage-with-unity-catalog.html, 2022.
\newblock last accessed: 13-Oct-22.

\bibitem{schule19Versioning}
Maximilian~E Schule, Lukas Karnowski, Josef Schmeißer, Benedikt Kleiner, Alfons Kemper, and Thomas Neumann.
\newblock {Versioning in Main-Memory Database Systems: From MusaeusDB to TardisDB}.
\newblock {\em Proceedings of the 31st International Conference on Scientific and Statistical Database Management}, pages 169--180, 2019.
\newblock \href {https://doi.org/10.1145/3335783.3335792} {\path{doi:10.1145/3335783.3335792}}.

\bibitem{seering12EffVer}
Adam Seering, Philippe Cudre-Mauroux, Samuel Madden, and Michael Stonebraker.
\newblock {Efficient Versioning for Scientific Array Databases}.
\newblock {\em 2012 IEEE 28th International Conference on Data Engineering}, 1:1013--1024, 2012.
\newblock \href {https://doi.org/10.1109/icde.2012.102} {\path{doi:10.1109/icde.2012.102}}.

\bibitem{si2022cost}
Wen Si, Li~Pan, and Shijun Liu.
\newblock A cost-driven online auto-scaling algorithm for web applications in cloud environments.
\newblock {\em Knowledge-Based Systems}, 244:108523, 2022.

\bibitem{simmhan2005survey}
Yogesh~L Simmhan, Beth Plale, Dennis Gannon, et~al.
\newblock A survey of data provenance techniques.

\bibitem{kmedianInapproximable}
Roberto Solis-Oba.
\newblock {\em Approximation Algorithms for the k-Median Problem}, pages 292--320.
\newblock Springer Berlin Heidelberg, Berlin, Heidelberg, 2006.
\newblock \href {https://doi.org/10.1007/11671541_10} {\path{doi:10.1007/11671541_10}}.

\bibitem{stivala2010lock}
Alex Stivala, Peter~J Stuckey, Maria~Garcia de~la Banda, Manuel Hermenegildo, and Anthony Wirth.
\newblock Lock-free parallel dynamic programming.
\newblock {\em Journal of Parallel and Distributed Computing}, 70(8):839--848, 2010.

\bibitem{suel2002zdelta}
Dimitre Trendafilov Nasir Memon~Torsten Suel.
\newblock zdelta: An efficient delta compression tool.
\newblock 2002.

\bibitem{vogel2020mosaic}
Lukas Vogel, Viktor Leis, Alexander van Renen, Thomas Neumann, Satoshi Imamura, and Alfons Kemper.
\newblock Mosaic: a budget-conscious storage engine for relational database systems.
\newblock {\em Proceedings of the VLDB Endowment}, 13(12):2662--2675, 2020.

\bibitem{Wang18forkbase}
Sheng Wang, Tien Tuan~Anh Dinh, Qian Lin, Zhongle Xie, Meihui Zhang, Qingchao Cai, Gang Chen, Beng~Chin Ooi, and Pingcheng Ruan.
\newblock Forkbase: An efficient storage engine for blockchain and forkable applications.
\newblock {\em Proc. VLDB Endow.}, 11(10):1137–1150, jun 2018.
\newblock \href {https://doi.org/10.14778/3231751.3231762} {\path{doi:10.14778/3231751.3231762}}.

\bibitem{xia14Ddelta}
Wen Xia, Hong Jiang, Dan Feng, Lei Tian, Min Fu, and Yukun Zhou.
\newblock {Ddelta: A deduplication-inspired fast delta compression approach}.
\newblock {\em Performance Evaluation}, 79:258--272, 2014.
\newblock \href {https://doi.org/10.1016/j.peva.2014.07.016} {\path{doi:10.1016/j.peva.2014.07.016}}.

\bibitem{Ying20pensive}
Tangwei Ying, Hanhua Chen, and Hai Jin.
\newblock Pensieve: Skewness-aware version switching for efficient graph processing.
\newblock In {\em Proceedings of the 2020 ACM SIGMOD International Conference on Management of Data}, SIGMOD '20, page 699–713, New York, NY, USA, 2020. Association for Computing Machinery.
\newblock \href {https://doi.org/10.1145/3318464.3380590} {\path{doi:10.1145/3318464.3380590}}.

\bibitem{simpleFollow-up}
Yin Zhang, Huiping Liu, Cheqing Jin, and Ye~Guo.
\newblock Storage and recreation trade-off for multi-version data management.
\newblock In Yi~Cai, Yoshiharu Ishikawa, and Jianliang Xu, editors, {\em Web and Big Data - Second International Joint Conference, APWeb-WAIM 2018, Macau, China, July 23-25, 2018, Proceedings, Part {II}}, volume 10988 of {\em Lecture Notes in Computer Science}, pages 394--409. Springer, 2018.
\newblock \href {https://doi.org/10.1007/978-3-319-96893-3\_30} {\path{doi:10.1007/978-3-319-96893-3\_30}}.

\end{thebibliography}

\appendix
\section{Approximation Algorithms}
\label{appendix:approx-alg-def}
We hereby define the notion of approximation algorithms used in this paper.

\begin{definition}
[$\rho$-approximation algorithm]
Let $\mathcal P$ be a minimization problem where we want to come up with a feasible solution $x$ satisfying some constraints (e.g., $a\cdot x \leq b$). We say that an algorithm $\mathcal A$ is a $\rho$-approximation algorithm for $\mathcal P$ if $x_{\mathcal A}$, the solution produced by $\mathcal A$ is feasible and that $OPT \leq f(x_{\mathcal A}) \leq \rho \cdot OPT$ where $OPT$ is an optimal objective value and $f(x)$ is the objective value of a solution $x$. Here, $\rho$ is the \emph{approximation ratio}.
Generally, we want $\mathcal A$ to run in polynomial time. 
\end{definition}

\begin{definition}[Polynomial-time approximation scheme (PTAS)]
A polynomial-time approximation scheme is an algorithm $\mathcal A$ that, when given any fixed $\epsilon>0$, can produce an $(1+\epsilon)$-approximation in time that is polynomial in the instance size. We say that $\mathcal A$ is a \emph{fully polynomial-time approximation scheme (FPTAS)} if the runtime of $\mathcal A$ is polynomial in both the instance size and $1/\epsilon$. 
\end{definition}

\begin{definition}[Bi-criteria approximation]
In problems such as ours where optimizing an objective function while meeting all constraints is challenging, we can consider relaxing both aspects. We say that an algorithm $\mathcal A$ $(\alpha,\beta)$-approximates problem $\mathcal P$ if the objective value of its output is at most $\alpha$ times the objective value of an optimal solution \textbf{and} the constraints are violated at most $\beta$ times.\footnote{We allow $x \leq \beta y$ if the constraint $x \leq y$ is presented.}
\end{definition}

\section{Optimization Problems with Known Hardness Results}\label{appendix:known-hardness-problems}
We hereby define a few problems with known hardness results that reduce to one of the versioning problems. 

\label{subsec:hard-problems}
Before we show our hardness results, it is useful to introduce several other \NP-hard problems to reduce from. 
\begin{definition}[\textsc{Set Cover}]
    Elements $U=\{o_1,\ldots,o_n\}$ and subsets $S_1,\ldots,S_m \subseteq U$ are given. 
    The goal is to find $A\subseteq [m]$ with minimum cardinality such that $\bigcup_{i\in A}S_i = U$. 
\end{definition}
\textsc{Set Cover} has no $c\ln n$-approximation for any $c<1$, unless $\textsc{P} = \NP$~\cite{setCoverNPhard}. 

\begin{definition}[\textsc{Subset Sum}]
Given real values $a_1,\ldots,a_n$ 
and a target value $T$. 
The goal is to find $A\subseteq [n]$ such that $\sum_{i\in A} a_i$ is maximized but not greater than $T$. 
\end{definition}
\textsc{Subset Sum} is also \NP-hard, but its FPTAS is well studied \cite{subsetSumFPTAS1,subsetSumFPTAS2,subsetSumFPTAS3,subsetSumFPTAS4,subsetSumFPTAS5}.

\begin{definition}[\textsc{k-median} and \textsc{Asymmetric k-median}]
    Given nodes $V=\{1,\ldots,n\}$, $k$, and symmetric (resp. asymmetric) distance measures $D_{i,j}$ for $i,j\in V$ that satisfies triangle inequality. The goal is to find a set of nodes $A\subseteq V$ of cardinality at most $k$ that minimizes 
    \[\sum_{v\in V}\min_{c\in A} D_{v,c}.\]
\end{definition}
The symmetric problem is well studied. The best known approximation lower bound for this problem is $1+\frac{1}{e}$. We note that an inapproximability result of $1+\frac{2}{e}$ \cite{symmetric-k-median-hardness-correct} is often mistakenly quoted for this problem, whereas the authors actually studied the $k$-median variant where the ``facilities'' and ``clients'' are in different sets. With the same method we can only get the hardness of $1+1/e$ in our definition. 

The asymmetric counterpart is rarely studied. The manuscript ~\cite{unpublishedArcher} showed that there is no $(\alpha,\beta)$-approximation ($\beta$ is the relaxation factor on $k$) if $\beta\leq \frac{1}{2}(1-\epsilon)(\ln n-\ln \alpha-O(1))$, unless $\NP\subseteq\DTIME(n^{O(\log\log n)})$. 

Notably, even symmetric $k-$median is inapproximable when triangle inequality is not assumed on the distance measure $D$. \cite{kmedianInapproximable} However, this hardness is not preserved by the standard reduction to \MSR{} (as in \cref{subsubsec:MSRMMRhardness}), since the path distance on graphs inherently satisfies triangle inequality. 

\begin{definition}[\textsc{k-center} and \textsc{Asymmetric k-center}]
    Given nodes $V=\{1,\ldots,n\}$, $k$, and asymmetric distance measures $D_{i,j}$ for $i,j\in V$ that satisfies triangle inequality. The goal is to find a set of nodes $A\subseteq V$ of cardinality at most $k$ that minimizes 
    \[\max_{v\in V}\min_{c\in A}D_{v,c}.\]
\end{definition}
The symmetric problem has a greedy 2-approximation, which is optimal unless P$=\NP$ \cite{SymmetrickCenterHardness}.

The asymmetric variant has $\log^* k$-approximation algorithms~\cite{k-center-approx}, and one cannot get a better approximation than $\log^* n$ unless $\NP\subseteq \DTIME(n^{O(\log\log n)})$, if we allow $k$ to be arbitrary \cite{k-center-hardness}.

\section{Reduction from General Trees to Binary Trees (Section \ref{subsec:tree-dp})}\label{appendix:WLOG-binary-tree}
\begin{lemma}
If algorithm $\calA$ solves \BMR{} on binary tree instances in $O(f(n))$ time where $n$ is the number of vertices in the tree, then there exists algorithm $\calA'$ solving \BMR{} on all tree instances in $O(f(2n))$ time. 
\end{lemma}
\begin{proof}[Proof Sketch]
If a node $v$ has more than two children, we modify the graph as follows:
\begin{bracketenumerate}
    \item Create node $v'$ and attach it as a child of $v$.
    \item Move all but the left-most children of $v$ to be children of $v'$
    \item Set the deltas of $(v,v') = (v',v) = 0$; set $(v',c_i) = (v, c_i)$ and $(c_i,v') = (c_i,v)$ for all transferred children $c_i$.
\end{bracketenumerate}
By repeating this process we obtain a binary tree with $\leq 2n$ nodes which has the same optimal objective value as before. Hence, after producing a binary tree, we can utilize the algorithm for binary tree to solve \BMR{} on any tree. 
\end{proof}


\begin{figure}[h!]
\begin{align*}
    S_1 &= s_v \\& 
    + \min_{\rho_1 + \rho_2 = \rho} \Big\{ 
    \min_{k_1, \gamma_1} \{DP[c_1][k_1][\gamma_2][\rho_1]\} 
    + \min_{k_2, \gamma_2} \{DP[c_2][k_2][\gamma_2][\rho_2]\Big\} \\
    S_2 & = s_v + s_{v,c_1} - s_{c_1}
    \\& +\min_{\rho_1 \le \rho}\Big\{DP[c_1][k-1][0][\rho_1 - (k-1) r_{v,c_1}] 
    + \min_{k',\gamma_2}\{ DP[c_2][k'][\gamma_2][\rho-\rho_1] \}\Big\}\\
    S_3 &= s_v + s_{v,c_2} - s_{c_2} \\& 
    + \min_{\rho_1 + \rho_2 = \rho}\Big\{\min_{k',\gamma_1}\{ DP[c_1][k'][\gamma_1][\rho_1] \}+ DP[c_2][k-1][0][\rho_2 - (k-1) r_{v,c_2}] \Big\} \\
    S_4 &= s_v + s_{v,c_1} - s_{c_1}+ s_{v,c_2} - s_{c_2} \\&
    + \min_{\rho_1 + \rho_2 = \rho}\min_{k_1 + k_2 = k-1}\Big\{ DP[c_1][k_1][0][\rho_1 - k_1 r_{v,c_1}] + DP[c_2][k_2][0][\rho_2 - k_2 r_{v,c_2}]\Big\} \\
    S_5 &= s_{c_1,v}
    \\& +\min_{\rho_1 \le \rho}\Big\{\min_{k_1}\{DP[c_1][k_1][\gamma - r_{c_1,v}][\rho_1 - \gamma]\} + \min_{k_2,\gamma'}\{ DP[c_2][k_2][\gamma'][\rho -\rho_1] \}\Big\}\\
    S_6 &= s_{c_2,v} \\&
    + \min_{\rho_1 + \rho_2 = \rho}\Big\{ \min_{k_2}\{DP[c_2][k_2][\gamma - r_{c_2,v}][\rho_2 - \gamma]\}
    +  \min_{k_1,\gamma'}\{DP[c_1][k_1][\gamma'][\rho_1]\} \Big\} \\
    S_7 &= s_{c_2,v} + s_{v,c_1} - s_{c_1} \\&
    + \min_{\rho_1+\rho_2 = \rho} \Big\{ DP[c_1][k-1][0][\rho_1 - (k-1)\cdot (r_{v,c_1}+\gamma)]
    + \min_{k'}\{DP[c_2][k'][\gamma - r_{c_2,v}][\rho_2 - \gamma]\}\Big\} \\
    S_8 &= s_{c_1,v} + s_{v,c_2} - s_{c_2} \\&
    + \min_{\rho_1+\rho_2 = \rho} \Big\{ \min_{k'}\{DP[c_1][k'][\gamma - r_{c_1,v}][\rho_1 - \gamma]\}
    + DP[c_2][k-1][0][\rho_2 - (k-1)\cdot (r_2+\gamma)]\Big\}\\
\end{align*}
\label{appendix:All-8-cases-tree-DP}
\caption{Recursion formulas for all 8 types of connections, for DP-MSR on bidirectional trees. }
\end{figure}

\section{Integer Linear Program (ILP) for MSR}\label{appendix:ILP}

In the following formulation, we have integer variables $\{x_e\}$ representing how many $v\in V$ is retrieved through the edge $e$. $I_e$ is a Boolean variable denoting whether edge $e$ is stored. 

We work on the extend graph with an auxiliary node: We use $V = V' \cup \{ v_0 \}$ where $V'$ are the versions and $v_0$ an auxiliary node. The materialization of any $v \in V'$ is represented by storing an edge $(v_0, v)$ with storage cost $s_v$ and retrieval cost $0$. 

\begin{equation*}
\begin{array}{rlllr}
    \min & \displaystyle\sum_{e \in E} r_e x_e  & & &\text{s.t.} \\
    & x_e\leq \vert V-1\vert I_e &&&\text{(indicator constraint)}\\
    & \displaystyle\sum_{e \in E} s_e I_e & \leq \calR & & \text{(storage cost)}\\
    & \displaystyle\sum_{e \in In(u)} x_e  &= \displaystyle\sum_{e \in Out(u)} x_e + 1 & ~\forall u \in V \setminus \{v_{aux}\} &\text{(sink)} \\
    & x_e &\in \{0,1,\ldots,\vert V\vert\}\\
    & I_e &\in \{0,1\}
\end{array}
\end{equation*}

\end{document}